\documentclass{article}

\usepackage[normalem]{ulem}

\usepackage{csquotes}
\usepackage{comment}
\usepackage{arxiv}
\usepackage{amssymb}
\usepackage{amsmath}
\usepackage{mathabx}
\usepackage{amsthm}
\usepackage{graphicx}
\usepackage[square,numbers]{natbib}
\usepackage{lipsum}
\usepackage{enumitem}
\usepackage{tikz}
\usepackage{caption}

\usepackage{xcolor}

\usepackage{bookmark}

\newtheorem{lemma}{Lemma}
\newtheorem{corollary}{Corollary}

\newtheorem*{conjecture*}{Conjecture}
\newtheorem*{lemma*}{Lemma}

\newtheorem{theorem}{Theorem}
\newtheorem{proposition}{Proposition} 

\newtheorem{definition}{Definition}

\theoremstyle{remark}
\newtheorem{remark}{Remark}

\title{A complete characterization of pairs of binary phylogenetic trees with identical $A_k$-alignments}

\renewcommand{\shorttitle}{Characterization of pairs of trees with identical $A_k$-alignments}

\hypersetup{
pdftitle={A complete characterization of pairs of binary phylogenetic trees with identical \texorpdfstring{$A_k$-alignments}{Ak-alignments}},
pdfauthor={Mirko Wilde, Mareike Fischer},
pdfkeywords={maximum parsimony, phylogenetic tree, Buneman theorem}
}

\author{
  Mirko Wilde\\
  \texttt{Institute of Mathematics and Computer Science, University of Greifswald, Germany}
  \\
  \texttt{mirko.wilde@uni-greifswald.de}
  \and
  \textbf{Mareike Fischer}\\
  \texttt{Institute of Mathematics and Computer Science, University of Greifswald, Germany}\\
  \texttt{mareike.fischer@uni-greifswald.de}\\\texttt{email@mareikefischer.de}
}

\begin{document}
\maketitle

\begin{abstract}{Phylogenetic trees play a key role in the reconstruction of evolutionary relationships. Typically, they are derived from aligned sequence data (like DNA, RNA, or proteins) by using optimization criteria like, e.g., maximum parsimony (MP). It is believed that the latter is able to reconstruct the \enquote{true} tree, i.e., the tree that generated the data, whenever the number of substitutions required to explain the data with that tree is relatively small compared to the size of the tree (measured in the number $n$ of leaves of the tree, which represent the species under investigation). However, reconstructing the correct tree from any alignment first and foremost requires the given alignment to perform differently on the \enquote{correct} tree than on others.

A special type of alignments, namely so-called $A_k$-alignments, has gained considerable interest in recent literature. These alignments consist of all binary characters (\enquote{sites}) which require precisely $k$ substitutions on a given tree. It has been found that whenever $k$ is small enough (in comparison to $n$), $A_k$-alignments uniquely characterize the trees that generated them. However, recent literature has left a significant gap between $n\leq 2k+2$  -- namely the cases in which no such characterization is possible -- and $n\geq 4k$  -- namely the cases in which this characterization works. It is the main aim of the present manuscript to close this gap, i.e., to present a full characterization of all pairs of trees that share the same $A_k$-alignment. In particular, we show that indeed every binary phylogenetic tree with $n$ leaves is uniquely defined by its $A_k$-alignments if $n\geq 2k+3$. By closing said gap, we also ensure that our result is optimal.}
\end{abstract}



\section{Introduction}\label{sec1}

Phylogenetic trees play a key role in the investigation of evolution \cite{Felsenstein2004, Semple2003}. They can be reconstructed and evaluated using different methods and criteria. One of the most famous such criteria is maximum parsimony (MP) \cite{Fitch1971}. However, just as other criteria, MP does not always recover a unique tree -- or, in other words, there are data (e.g., DNA alignments) for which several trees are equally good and thus equally likely to be the \enquote{true} tree, i.e., the tree that generated the data. Understanding in which cases there is a unique MP tree is of high relevance as a first step to reconstruct the correct tree. However, in order for any method to be able to reconstruct the \enquote{correct} tree from a sequence alignment, this alignment has to perform differently on that tree than on others. In terms of MP, it is necessary that the alignment's parsimony score on said tree differs from that on other trees.

In the present manuscript, we analyze a specific type of data, namely so-called $A_k$-alignments, which contain all binary characters (e.g., functions that attribute 0's and 1's to all species under investigation depending on whether or not a certain characteristic is present in the species or not) that require precisely $k$ substitutions on a given tree $T$ -- i.e., these characters have \emph{parsimony score} $k$ on $T$. Said alignments have gained considerable interest in the recent literature as it has been shown that when $k$ is small enough (compared to the number $n$ of species under investigation), $A_k$-alignments characterize the underlying tree. In particular, this has been shown to be the case whenever $n \geq 20k$ \cite{Fischer2023}), and this bound was later improved to $n \geq 4k$ \cite{WildeFischer2023}. Moreover, it has been shown that whenever $n\leq 2k+2$, there are cases in which two trees share the same $A_k$-alignment, making a characterization of individual trees based on their $A_k$-alignments impossible in said cases. Thus, there is a gap in the literature for values of $n$ in the set $\{2k+3,\ldots,4k-1\}$, for which it is not known if the $A_k$-alignment of a tree characterizes this tree. 

In the present manuscript, we first close the aforementioned gap. In particular, we show that for all $n\geq 2k+3$, the $A_k$-alignment of a tree $T$ uniquely characterizes $T$. Note that the fact that we close the gap also implies that our result is optimal in the sense that the bound of $n \geq 2k+3$ cannot be further improved. 

The closed gap can be viewed as an important first step towards showing the reliability of MP as a tree reconstruction criterion whenever the number of changes is small enough, which has been frequently observed in data already \cite{Sourdis1988}. 

The second main aim of our manuscript is then to characterize pairs of (non-isomorphic) trees which share the same $A_k$-alignment and whose leaf number is smaller than $2k+3$. This will shed further light on the \enquote{problematic} case in which different trees with the same $A_k$-alignment.

\section{Preliminaries}

\subsection{Terminology}
We adopt in large parts the terminology which we already used in \cite{WildeFischer2023}. As our results lie in the intersection of phylogenetic combinatorics and graph theory, we need basic concepts from both fields.

\subsubsection*{Basic concepts from phylogenetic combinatorics}
A simple connected graph $T=(V,E)$ is called a \emph{tree} if $T$ does not contain a cycle. The vertices in $T$ with degree at most $1$ are called \emph{leaves}. A \emph{phylogenetic $X$-tree} $T=(V,E)$ is a tree such that no vertex has degree $2$ and such that the leaves of the tree are bijectively labelled by and thus in the following also identified with $X$. The finite set $X$ is often referred to as \emph{species set} or \emph{taxon set}. Whenever there is no ambiguity concerning $X$ or whenever $X$ is irrelevant, we often write \emph{phylogenetic tree} instead of phylogenetic $X$-tree. Moreover, unless stated otherwise, we may assume $X=\{1,\ldots,n\}$. A phylogenetic $X$-tree with maximum degree $3$ is called \emph{binary}. A \emph{rooted} phylogenetic $X$-tree $T$ is a tree with at most one vertex of degree 2 and whose leaf set is bijectively labelled by $X$ and in which one inner vertex $\rho$ is declared to be the \emph{root}. In case such a rooted phylogenetic $X$-tree contains a degree-2 vertex, this vertex has to be the root. Moreover, the special case in which a tree consists only of one vertex, for technical reasons this tree is also considered to be a rooted binary phylogenetic $X$-tree with its only vertex being at the same time the root and its only leaf. Note that whenever $\vert X \vert \geq 2$, a rooted binary phylogenetic $X$-tree can be subdivided into its two \emph{maximal pending subtrees}, i.e., the subtrees adjacent to the root, cf. trees $T_A$ and $T_B$ in and their maximal pending subtrees depicted in Figure \ref{fig:treedecomp}.

As the graph theoretical isomorphism concept is not strong enough to capture evolutionary relationships, the concept of a \emph{phylogenetic isomorphism} is stricter than its graph theoretical analog in the following sense: Let $T=(V,E)$ and $\widetilde{T}=(\widetilde{V},\widetilde{E})$ be two phylogenetic $X$-trees which are isomorphic as graphs, i.e., there exists an bijection $\phi: V \rightarrow \widetilde{V}$ such that $\{v,\widetilde{v}\}\in E(T) \Longleftrightarrow \{\phi(v), \phi(\widetilde{v})\}\in E(\widetilde{T})$. If additionally $\phi(x) = x$ holds for all $x\in X$, and additionally -- in the rooted case -- if $\phi(\rho_T)=\rho_{\widetilde{T}}$ for the roots $\rho_T$ of $T$ and $\rho_{\widetilde{T}}$ of $\widetilde{T}$, then $\phi$ is called an \emph{isomorphism of (rooted) phylogenetic $X$-trees}. Whenever we talk about \emph{isomorphic} phylogenetic $X$-trees, we mean this to imply the existence of such an isomorphism. If $T$ and $\widetilde{T}$ are isomorphic phylogenetic $X$-trees, we denote this by $T\cong \widetilde{T}$.

 \subsubsection*{Basic concepts from graph theory}
 As some of our results are inspired by Menger's famous theorem, we now  introduce some concepts from this context. Let $G=(V,E)$ be a simple graph and let $P= v_1, \ldots, v_k$  (for $k\in \mathbb{N}_{\geq 1}$) be a sequence of vertices in $V$. If the elements of this sequence are pairwise distinct and if $\{v_i, v_{i+1}\}\in E$ for $i\in \{1, \ldots, k-1\}$, then $P$ is called a \emph{path}. Then $v_1, v_k$ are called the \emph{endpoints} of the path, and in the case $k\geq 3$ we call $v_2, \ldots, v_{k-1}$ the \emph{interior vertices} of the path. We do not distinguish between $v_1, v_2, \ldots, v_{k-1}, v_k$ and $v_k, v_{k-1}, \ldots, v_2, v_1$, so these sequences are considered to denote the same path. Now let $A,B$ be two subsets of $V$. If $P=v_1, \ldots, v_k$ is a path and $A\cap \{v_1, \ldots, v_k\} = \{v_1\}, B\cap \{v_1, \ldots, v_k\} = \{v_k\}$, we call $P$ an $A$-$B$-path. Of course $P$ is an $A$-$B$-path if and only if it is an $B$-$A$-path. In the case that one of the sets $A$ or $B$ is a singleton $\{v\}$, we sloppily drop the brackets. In particular, we can speak of a \emph{$v$-$A$-path}, of an \emph{$A$-$w$-path} or of a \emph{$v$-$w$-path} if $v$ and $w$ are vertices of $G$ and $A$ is a subset of $V$. If $T=(V,E)$ is a phylogenetic $X$-tree and $P= v_1, \ldots, v_k$ with $k\geq 2$, $v_2,\ldots,v_{k-1} \in V$ and $v_1, v_k\in X$, then $P$ is called a \emph{leaf-to-leaf-path}.

\subsubsection*{Characters, $X$-splits and alignments}
One goal of phylogenetics is the reconstruction of phylogenetic $X$-trees from data (such as binary character data). Thus, we now introduce the concepts which are used for modelling such data.
A function $f: X \rightarrow \{a,b\}$ is called a \emph{binary character}. Moreover, if $g: V(T) \rightarrow \{a,b\}$ is a function with the property that ${g}{\vert_ X}=f$, i.e., with the property that the restriction of $g$ on $X$ coincides with $f$, then $g$ is called an \emph{extension} of $f$ on $T$. Let $ch(g,T) = \{e=\{v,w\}\in E(T): g(v) \neq g(w)\}$ be the set of \emph{changing edges} of $g$ on $T$. Then $|ch(g,T)|$ is called the \emph{changing number} of $g$ on $T$. For a given binary character $f: X\rightarrow \{a,b\}$ and a given phylogenetic $X$-tree $T=(V,E)$,  let $\widetilde{g}$ be an extension of $f$ with the minimal changing number. Then we define $l(f,T)$ as the changing number of $\widetilde{g}$. In this case $l(f,T)$ is called the \emph{parsimony score} of $f$ on $T$, and $\widetilde{g}$ is called a \emph{most parsimonious} or \emph{minimal} extension of $f$ on $T$.

There is an important connection between binary characters and another concept which will turn out to be very important for our paper, namely so-called \emph{$X$-splits}:  Let $\emptyset \neq A,B\subset X$ such that $X=A\cup B$ is a bipartition. Then the unordered pair $A,B$ is called an \emph{$X$-split} and is denoted by $\sigma = A\vert B$. If $f: X \rightarrow \{a,b\}$ is binary character with $\{a,b\} \subseteq f(X)$ (i.e., $f$ employs both character states $a$ and $b$), then let $\emptyset\neq A_f = f^{-1}(\{a\})$ and $\emptyset\neq B_f=f^{-1}(\{b\})$. In this case, $A_f\vert B_f$ is an $X$-split and is called the $X$-split \emph{induced by $f$}. We note that for each $X$-split $A\vert B$ there are exactly two binary characters which induce it (as the roles of $A$ and $B$ can be interchanged).  

There is another reason why $X$-splits play an important role in phylogenetics: They can be used to model the edges of an phylogenetic $X$-tree. Let $T = (V,E)$ be a phylogenetic $X$-tree and $e\in E$. Then $e$ defines a unique $X$-split $A_e\vert B_e$ such that for each $x,y\in X$ we have that $|A_e\cap \{x,y\}| = |B_e\cap \{x,y\}| = 1$ if and only if $e$ is contained in the unique $x$-$y$-path of $T$. In this case, $A_e \vert B_e$ is called an $X$-split \emph{induced by $T$}. We use the notation $\Sigma(T)$ for the set of all $X$-splits induced by $T$. Now let $\sigma$ be an $X$-split with $\sigma = A\vert B$ and $\lvert A\rvert \leq \lvert B\rvert$. Then we define $\lvert \sigma \rvert = \lvert A \rvert$ and call this value the \emph{size} of $\sigma$. If $\lvert \sigma \rvert = 1$, then $\sigma$ is called \emph{trivial}. Note that each edge which contains a leaf induces a trivial split. This implies that for a given set $X$, each trivial $X$-split is contained in $\Sigma(T)$ for each phylogenetic $X$-tree. For this reason it is often convenient to consider $\Sigma^\ast(T)$, which is defined as the set of all non-trivial $X$-splits induced by $T$. Equivalently, $\Sigma^\ast(T)$ can be characterized as the set of all $X$-splits induced by edges of $T$ which are not incident to any leaf. 

Note that two binary characters $\sigma_1=A_1\vert B_1$ and $\sigma_2= A_2 \vert B_2$ on a taxon set $X$ are called \emph{compatible} if at least one of the four intersections $A_1\cap A_2$, $A_1\cap B_2$, $B_1\cap A_2$ or $B_1 \cap B_2$ is empty.  

It is not hard to see that there is a correlation between $X$-splits induced by a phylogenetic $X$-tree $T$ and binary characters on $X$ with parsimony score $1$ on $T$. If $l(f,T) = 1$, then there is a unique most parsimonious extension $g$ of $f$ which has exactly one changing edge $e$. In this case, the $X$-split $A_f\vert B_f$ induced by $f$ is the same as the $X$-split induced by $e$. If we denote by $A_1(T)$ the set of all binary characters $f$ with parsimony score $1$ on $T$, then the set of all $X$-splits induced by characters in $A_1(T)$ is precisely the set $\Sigma(T)$. Therefore, for a given phylogenetic $X$-tree $T$ it seems to be a natural generalization to define $A_k(T)$ as the set of all binary characters $f$ on $X$ with parsimony score $k$ on $T$. If $A$ is a set of binary characters on $X$, then we call $A$ an \emph{$A_k$-alignment} if there is a phylogenetic $X$-tree $T$ with $A = A_k(T)$. Note that $A_k$-alignments are the most important concept in our manuscript, as we seek to answer the question of when these character sets uniquely define a tree. 

\subsubsection*{Phylogenetic tree operations}
Throughout our manuscript, we will require certain phylogenetic tree operations as tools in our proofs. There are three concepts which will be particularly important in this regard: The \emph{nearest neighbor interchange (NNI) moves}, \emph{cherry reduction} and \emph{attaching a leaf/cherry}.

 If there is some inner edge $e = \{v,w\}$ of a binary phylogenetic $X$-tree $T$, there are two \emph{NNI moves induced by $e$} defined in the following way: Let $e_1, e_2, e_3, e_4$ be the four edges adjacent to $e$ in $T$. We can assume that $v$ is incident with $e_1$ and $e_2$ and that $w$ is incident with $e_3, e_4$. Let $e_1 = \{v, v_1\}, e_2 = \{v,v_2\}, e_3 = \{w, w_1\}$ and $e_4 = \{w, w_2\}$. Then, a new binary phylogenetic $X$-tree $T_1$ can be constructed by deleting $e_1$ and $e_3$ from $T$ and inserting new edges $f_1 = \{v,w_1\}, f_2 = \{w,v_1\}$ (cf. Figure \ref{fig:NNI}). Analogously, another binary phylogenetic $X$-tree $T_2$ can be constructed by deleting $e_1$ and $e_4$ from $T$ and inserting new edges $f_3 = \{v,w_2\}, f_4 = \{w, v_1\}$. We call the elements of $\{T_1,T_2\}$ the \emph{NNI neighbors} of $T$ and the set $\{T, T_1, T_2\}$ the \emph{NNI neighborhood} of $T$ (i.e., $T$ is part of its own neighborhood, even though it is not its own neighbor). 

\begin{figure}
    \centering
\includegraphics{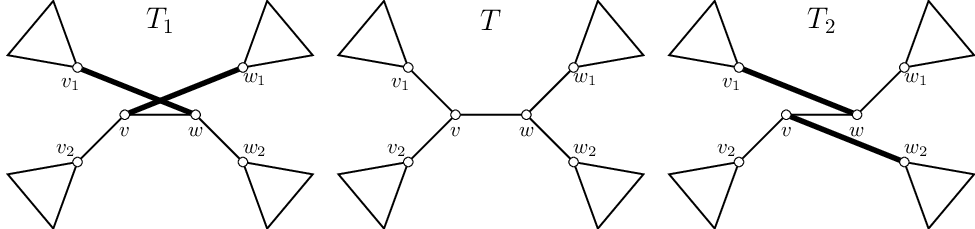}
\caption{NNI neighborhood of $T$.}
\label{fig:NNI}
\end{figure}

Next we define two types of \emph{cherry reductions}. Let $T$ be a binary phylogenetic $X$-tree with $\vert X \vert >1$. For each $x\in X$, let $v_x$ be the unique vertex in $T$ which is adjacent to $x$. If $x,y\in X$ and $v_x = v_y$, then we call $[x,y]$ a \emph{cherry} of $T$. We will consider the following options to construct a new binary phylogenetic $X$-tree. 
\begin{enumerate}
    \item A  \emph{cherry reduction of type 1 using $[x,y]$} works as follows: Delete $x,y$. Then relabel the resulting new leaf $v_x$ by $y$. This operation induces a phylogenetic $X\setminus \{x\}$-tree which we will denote by $T^1$.
    \item A \emph{cherry reduction of type 2 using $[x,y]$} works as follows: Delete $x,y$ and $v_x$. Then suppress the resulting degree-$2$ vertex. This operation induces a phylogenetic $X\setminus \{x,y\}$-tree which we will denote by $T^2$.
\end{enumerate}

 We remark that $T^1$ and $T^2$ are binary if $\lvert X\rvert \geq 3$. Cherry reductions are often used in inductive proofs because they reduce the number of leaves. There are two facts which contribute to the usefulness of cherry reductions in this context: First, every binary phylogenetic $X$-tree with at least three vertices contains at least one cherry. Second, every binary phylogenetic tree with at least four vertices contains at least two disjoint cherries  \cite[Proposition 1.2.5]{Semple2003}.

The last tree operation we want to introduce is that of \emph{attaching a leaf or cherry}. This can be regarded as an inverse operation to cherry reductions in the sense that one can obtain $T^1$ and $T$ from $T^2$ in the following way: There is a unique edge $e\in E(T^2)\setminus E(T)$ resulting from the last step in a cherry reduction of type 2 in which a degree-2 vertex is suppressed. So by starting with $T^2$, then subdividing $e$ by introducing a degree-2 vertex $v$, we get the suppressed degree-2 vertex back. Then we add $y$ to the taxon set of $T^2$ and insert an edge $\{v,y\}$. In this way, we obtain $T^1$ from $T^2$. We call this operation \emph{attaching leaf $y$ to edge $e$}. Now we repeat this operation by attaching leaf $x$ to $\{v,y\}$. By the  concatenation of these steps we ultimately obtain $T$ from $T^2$. We call this sequence of operations \emph{attaching cherry $[x,y]$ to $e$}. So if we first apply a cherry reduction of type 2 to $T$ and then attach the deleted cherry to the unique element of $E(T^2)\setminus E(T)$, we have $T$ again. On the other hand, if we first attach a cherry $[x,y]$ to some binary phylogenetic $X\setminus \{x,y\}$-tree $T'$ and then apply a cherry reduction of type 2 using $[x,y]$, we have $T'$ again.

We conclude this section by noting that whenever we reduce a binary phylogenetic $X$-tree by a cherry reduction, we can reduce a given binary character $f:X\rightarrow \{a,b\}$ together with it by simply restricting it to the remaining taxon set. This is formalized in the following definition.

\begin{definition}[Definition 1 in \cite{WildeFischer2023}]\label{def:cherryred}
Let $T$ be a binary phylogenetic $X$-tree. Let $T^1$ and $T^2$ be the two trees derived from $T$ when performing a cherry reduction of type 1 or 2, respectively, to cherry $[x,y]$ of $T$. Moreover, let $f: X\to \{a,b\}$ be a binary character.
\begin{enumerate}
\item If $f(x)=f(y)$, we define $f^1$ as the restriction of $f$ to the taxa of $T^1$.
\item If $f(x)\neq f(y)$, we define $f^2$ as the restriction of $f$ to the taxa of $T^2$.
\end{enumerate}
\end{definition}

We are now in the position to state some known results from the literature which we will use subsequently.

\subsection{Known results}
As we have seen in the previous section, the parsimony score of a character is usually defined and analyzed in terms of minimal extensions over a given tree $T$. However, for our results knowing other methods for characterizing the parsimony score of a character will be useful, too. Thus, we will explain some results from the literature which are helpful in this regard. At first we will introduce a statement about the relationship between the parsimony score of a character and the existence of some specific collection of edge-disjoint leaf-to-leaf paths. This statement is a corollary of Menger's famous theorem. Then we will mention some results about cherry reductions and show how the calculation of a parsimony score for a character with respect to a given tree can be reduced to the calculation of the parsimony scores of restrictions of this character to the taxon sets of smaller trees.

\subsubsection*{Known results on leaf-to-leaf paths and their relation to maximum parsimony}

As there are various versions of Menger's theorem \cite{Diestel2017} and for the sake of completeness, we state the theorem of Menger in the version on which the results are based which we cite subsequently (\cite{Menger1927}, \cite[Theorem 3.3.1]{Diestel2017}).

\begin{theorem}[Menger's theorem \cite{Menger1927}]\label{thm:menger3}
Let $G$ be a graph with vertex set $V$ and $A,B\subset V$. Then the minimal cardinality of a subset $W\subset V$ such that $G \setminus W$ contains no $A$-$B$-path is equal to the maximum number of vertex-disjoint $A$-$B$-paths.
\end{theorem}

The following result can be derived as a direct consequence of Menger's theorem and is the standard variant of this theorem within the area of mathematical phylogenetics.

\begin{proposition}[Corollary 5.1.8 in \cite{Semple2003}]\label{prop:menger0}
Let $T$ be a binary phylogenetic $X$-tree, and let $f: X\rightarrow \{a,b\}$ be a binary character. Then, $l(f,T)$ is equal to the maximum number of edge-disjoint $A_f$-$B_f$-paths in $T$.
\end{proposition}

The following lemma, which was very useful in \cite{WildeFischer2023} for the proof of Theorem \ref{thm:bound4k}, is strongly related to Proposition \ref{prop:menger0}. It does not only tell us the maximal number of edge-disjoint leaf-to-leaf paths in a binary phylogenetic tree, but also guarantees that we can choose one endpoint of at least one such path.

\begin{lemma}\label{lem:leafchoice}(Lemma 2 in \cite{WildeFischer2023})
Let $T$ be a binary phylogenetic tree on taxon set $X$ with $\vert X \vert = n\geq 2$. Let $x\in X$. Then $T$ contains $p=\left\lfloor\frac{n}{2}\right\rfloor$ edge-disjoint leaf-to-leaf paths $P_1,\ldots,P_{p}$ such that $x$ is an endpoint of $P_1$. 
\end{lemma}

We will later state a variant of Lemma \ref{lem:leafchoice}, namely Lemma \ref{lem:leafchoice2}, which will be useful to prove the main statements of our manuscript.

Now, as already stated above, some of our proofs use inductive arguments. Thus, we need to  derive the parsimony score of some binary character $f$ on the taxon set $X$ a given tree $T$ from the parsimony scores of restrictions $f^1$ and $f^2$ of $f$ to the taxon sets of smaller trees. The main ingredient in achieving this is the following lemma.

\begin{lemma}[Lemma 3 and Lemma 4 in \cite{WildeFischer2023}]\label{lem:cherryred1} Let $k \in \mathbb{N}_{\geq 1}$, $n \in \mathbb{N}_{\geq 3}$ and let $X=\{1,\ldots,n\}$. Let $T$ be a binary phylogenetic $X$-tree with $|X|=n$ and a cherry $[x,y]$. 
Consider the partition $A_k(T) = A_1 \cup A^1_2\cup A^2_2$ with $A_1$ defined as the alignment containing all characters $f\in A_k(T)$ for which $f(x) = f(y)$, with $A^1_2$ defined as the alignment containing all characters $f\in A_k(T)$ for which $f(x) = a, f(y) = b$ and with $A^2_2$ defined as the alignment containing all characters $f\in A_k(T)$ for which $f(x) = b, f(y) = a$.
Let $T^1$ and $T^2$ be the two trees derived from $T$ when performing a cherry reduction of type 1 or 2, respectively, using cherry $[x,y]$. Consider the following maps given by Definition \ref{def:cherryred}:
\begin{enumerate}
\item If $f\in A_1$, then $f$ is mapped to $f^1$.
\item If $f\in A^i_2$, then $f$ is mapped to $\left(f^2, i\right)$ for $i = 1,2$.
\end{enumerate}
Then, the first map is a bijection from $A_1$ to $A_k(T^1)$ and the second map is a bijection from $A^1_2 \cup A^2_2$ to $A_{k-1}(T^2) \times \{1,2\}$.
\end{lemma}

There is an equivalent formulation of Lemma \ref{lem:cherryred1} in the language of matrices. As is often done for alignments in mathematical phylogenetics, for a given phylogenetic $X$-tree we can represent $A_k(T)$ as a matrix with $|X|$ rows and $|A_k(T)|$ columns. Each row represents an element of the taxon set and each column a binary character with parsimony score $k$. An entry $(i,j)$ of the matrix has value $a$ if the function represented by column $j$ maps the taxon represented by row $i$ to $a$, otherwise it has value $b$. Then the lemma says that we can order the columns and rows such that the first two rows represent the taxa of cherry $[x,y]$ and the matrix $A_k(T)$ can be represented as a composition of the matrix $A_k(T^1)$ and $A_{k-1}(T^2)$ as indicated in Table \ref{table_alignmentdecomp}.

\begin{table}\begin{center}$\underbrace{
\begin{tabular}{ |cccc| } 
 \hline
 $aaa$ $\cdots$& $bbb$ $\cdots$& $aaa$ $\cdots$ & $bbb$ $\cdots$
\\
 \cline{1-2}
 $aaa$ $\cdots$ & $bbb$ $\cdots$ & \multicolumn{1}{|c}{$bbb\cdots$} & $aaa$ $\cdots$\\ 
 \cline{3-4}
 \multicolumn{2}{|c|}{ \rule{0pt}{3ex}$A_k(T^1)$} & $A_{k-1}(T^2)$ & $A_{k-1}(T^2)$\\ 
 \hline
\end{tabular}
}_{\text{\scalebox{1.5}{$A_k(T)$}}}$ \end{center}\caption{Decomposition of alignment $A_k(T)$ in terms of $A_k(T^1)$ and $A_{k-1}(T^2)$.}\label{table_alignmentdecomp}
\end{table}

As a consequence of this matrix representation, $A_k(T^1)$ and $A_{k-1}(T^2)$ are mere submatrices of $A_k(T)$. From this we can immediately conclude the following corollary.

\begin{corollary}[adapted from Corollary 2 in \cite{WildeFischer2023}]\label{cor:cherryred} Let $k \in \mathbb{N}_{\geq 2}$. Let $T, \widetilde{T}$ be binary phylogenetic $X$-trees and let $[x,y]$ be a cherry contained in both $T$ and $\widetilde{T}$. Let $T^1$ and $\widetilde{T}^1$ as well as $T^2$ and $\widetilde{T}^2$ be the phylogenetic trees resulting from cherry reductions of types 1 and 2, respectively, using cherry $[x,y]$. Then, we have $A_k(T) = A_k(\widetilde{T})$ if and only if $A_k(T^1) = A_k(\widetilde{T}^1)$ and $A_{k-1}(T^2) = A_{k-1}(\widetilde{T}^2)$.
\end{corollary}

Note that Corollary \ref{cor:cherryred} was stated in \cite{WildeFischer2023} only for the special case that $T$ and $\widetilde{T}$ are NNI neighbors. However, the proof did not use this condition, which implies that the statement also holds for the more general case stated here.

\subsubsection*{Known results concerning tree characterizations and $A_k$-alignments}

The main purpose of our manuscript is the characterization of pairs of trees that share the same $A_k$-alignment. In this regard, some progress has already been made in the literature, and we will now summarize the current state of research. The first statement in this context is covering the cases $k=1$ and $k=2$. 
\begin{proposition}[adapted from Corollary 1 in \cite{Fischer2023} and Proposition 1 in \cite{Fischer2019}]\label{prop:A1A2} Let $k \in \{1,2\}$. Let $T$ and $\widetilde{T}$ be two binary phylogenetic $X$-trees. Then, $T \cong \widetilde{T}$ if and only if $A_k(T)=A_k(\widetilde{T})$.
\end{proposition}

The proofs of Proposition \ref{prop:A1A2} are both based on the classic theorem by Buneman \cite{Buneman1971} (also known as the Splits Equivalence Theorem \cite{Semple2003}), which we also need in the present manuscript.

\begin{theorem}[Buneman's Theorem, Splits Equivalence Theorem]\label{thm:buneman}
     Let $T$ and $\widetilde{T}$ be two binary phylogenetic $X$-trees. Then, $T \cong \widetilde{T}$ if and only if $\Sigma(T)=\Sigma(\widetilde{T})$. Moreover, for a set of $X$-splits $\Sigma$ there exists a unique phylogenetic tree $T$ with $\Sigma(T)=\Sigma$ if and only if the elements of $\Sigma$ are pairwise compatible.
\end{theorem}

The following result is useful when counting characters in $A_k$-alignments. 

\begin{theorem} \label{thm:lengthAk} \citep{Steel1993,Steel2016}
Let $T$ be a binary phylogenetic $X$-tree with $\vert X\vert =n$. Then, we have:
$$\vert A_k(T)\vert = \frac{2n-3k}{k}\binom{n-k-1}{k-1} \cdot 2^k.$$
\end{theorem}

The following recent result from the literature characterizes the pairs with identical $A_k$-alignments in the case where $n\geq 4k$. 

\begin{theorem}[Theorem 2 in \cite{WildeFischer2023}]\label{thm:bound4k} Let $k \in \mathbb{N}_{\geq 1}$. Let $T$ and $\widetilde{T}$ be two binary phylogenetic $X$-trees with $\lvert X \rvert = n \geq 4k$. Then, $T \cong \widetilde{T}$ if and only if $A_k(T)=A_k(\widetilde{T})$.
\end{theorem}

Moreover, in case that $T$ and $\widetilde{T}$ are NNI neighbors,  the following result concerning their $A_k$-alignments is known.

\begin{proposition}[Proposition 4 in \cite{WildeFischer2023}]\label{prop:nni}
Let $k \in \mathbb{N}_{\geq 2}$. Let $T$ be a binary phylogenetic $X$-tree with $n=\vert X\vert \geq 4$ and $A\vert B \in \Sigma^\ast(T)$ inducing the maximal pending subtrees $T_{A_1}$ and $ T_{A_2}$, whose leaves are subsets of $A$, as well as the maximal pending subtrees $T_{B_1}$ and $ T_{B_2}$, whose leaves are subsets of $B$, cf. Figure \ref{fig:treedecomp}. Let $n_1 = \vert A_1\vert$, $n_2 = \vert A_2\vert$, $n_3 = \vert B_1\vert$ and $n_4 = \vert B_2\vert$. Moreover, let $\widetilde{T}$ be the tree obtained from $T$ by exchanging $T_{A_2}$ with $T_{B_2}$ (i.e., $T$ and $\widetilde{T}$ are NNI neighbors). Let $s(T,\widetilde{T}) = \sum\limits_{i=1}^4 \left\lfloor \frac{n_i-1}{2}\right\rfloor$. Then, we have: 
\begin{align*} A_k(T) = A_k(\widetilde{T}) & \ \ \Longleftrightarrow \ \  
s(T,\widetilde{T}) < k-2.
\end{align*}
\end{proposition}

Last but not least, Proposition \ref{prop:nni} leads to the following corollary which is a special case of a more general result which we will prove in the present manuscript (namely Theorem \ref{thm:bound2k+3}).

\begin{corollary}[Theorem 3 in \cite{WildeFischer2023}] \label{cor:bound2k+3_nni}
Let $k\in \mathbb{N}_{\geq 1}$ and $n\geq 2k+3$. Let $T, \widetilde{T}$ be a pair of binary phylogenetic $X$-trees with $\lvert X\rvert = n$ such that $\widetilde{T}$ is in the NNI neighborhood of $T$. Then $A_k(T) = A_k(\widetilde{T})$ if and only if $T\cong \widetilde{T}$.    
\end{corollary}

\begin{figure} 
\center
\includegraphics[width=0.9\textwidth]{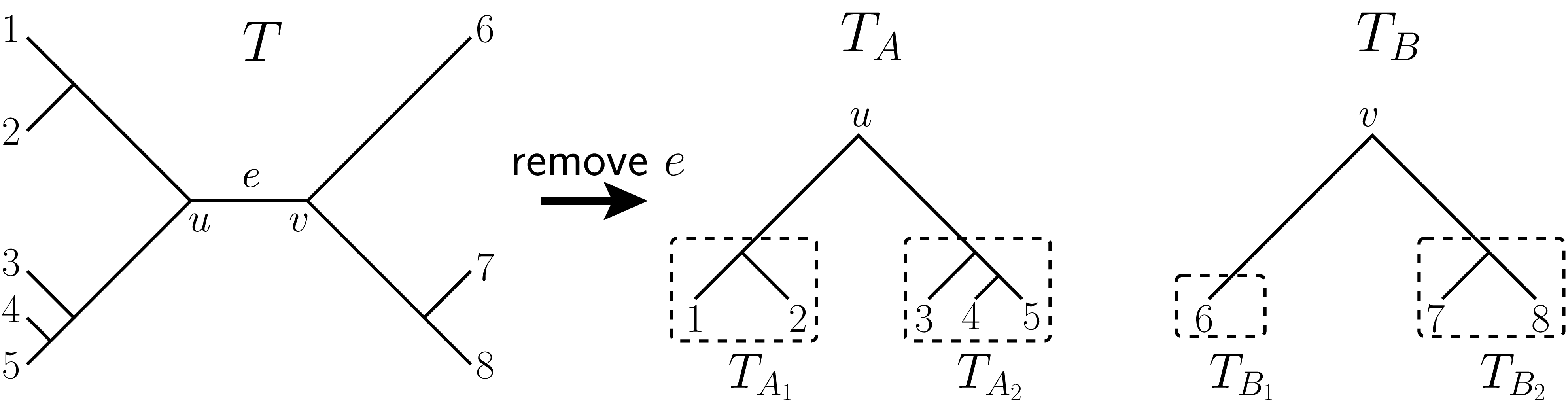}
\caption{  \scriptsize (taken from \cite{Fischer2019, Fischer2023}) By removing an edge $e=\{u,v\}$ from an unrooted phylogenetic tree $T$, it is decomposed into two subtrees, $T_A$ and $T_B$. If, as in this figure, both of them consist of more than one node, then we can further decompose them into their two maximal pending subtrees $T_{A_1}$ and $T_{A_2}$ or $T_{B_1}$ and $T_{B_2}$, by removing $u$ or $v$, respectively. 
}
\label{fig:treedecomp}
\end{figure}

As can be seen from all listed results, the characterization of pairs of binary phylogenetic trees with identical $A_k$-alignments is to-date incomplete. In particular, nothing is known for combinations of $k>2$ and $n<4k$ in case the trees under consideration are not NNI neighbors. The main goal of our manuscript is to fill this gap and give a complete characterization.

\section{Results}\label{sec:results}

The main purpose of our manuscript is twofold: First, we want to close the literature's gap concerning values of $n$ contained in the set $\{2k+3,\ldots,4k-1\}$. Note that by Theorem \ref{thm:bound4k} we know that for all values of $n$ larger than the ones contained in said set, we can guarantee that two trees are isomorphic if and only if their $A_k$-alignments coincide. Moreover, in \cite{WildeFischer2023} a construction for examples was presented with $n=2k+1$ and $n=2k+2$ in which two non-isomorphic trees have identical $A_k$-alignments. So this shows that improving the bound of $4k$ from Theorem \ref{thm:bound4k}  to $2k+3$ is best possible.

The second main aim of our manuscript is then to characterize pairs of (non-isomorphic) trees which share the same $A_k$-alignment and whose leaf number is smaller than $2k+3$. We will do so in Section \ref{sec:n_smaller}.

However, in Section \ref{sec:improveThm3} we start with the improvement of Theorem \ref{thm:bound4k} to $2k+3$ as explained above.

\subsection{Improvement of the bound given by Theorem \ref{thm:bound4k}: the case \texorpdfstring{$n\geq 2k+3$}{n>=2k+3}}\label{sec:improveThm3}

We now state our first main result, which improves the bound given by Theorem \ref{thm:bound4k} from $4k$ to $2k+3$. 

\begin{theorem} \label{thm:bound2k+3}Let $k \in \mathbb{N}_{\geq 1}$. Let $T$ and $\widetilde{T}$ be two binary phylogenetic $X$-trees with $\lvert X \rvert = n \geq 2k+3$. Then, $T \cong \widetilde{T}$ if and only if $A_k(T)=A_k(\widetilde{T})$.
\end{theorem}

It is the main aim of this section to prove Theorem \ref{thm:bound2k+3}. Before we can do so, we first remark that the bound given by Theorem \ref{thm:bound2k+3} is sharp as examples for $n\leq 2k+2$ are known in which non-isomorphic phylogenetic $X$-trees share the same $A_k$-alignment \cite{WildeFischer2023}. However, Theorem \ref{thm:bound2k+3} does not only generalize Theorem \ref{thm:bound4k} by improving the given bound, it also generalizes Corollary \ref{cor:bound2k+3_nni}, which only holds for NNI neighbors, to general pairs of trees.

We now turn our attention to some preliminary work which is necessary in order to prove Theorem \ref{thm:bound2k+3} subsequently. 

\subsubsection{Preparation for the proof of Theorem \ref{thm:bound2k+3}}

Recall that Lemma \ref{lem:leafchoice} states that we can fix a leaf and claim the existence of a collection of leaf-to-leaf paths of which one contains the fixed leaf as an endpoint. Unfortunately, while this lemma was the main ingredient in the proof of Theorem \ref{thm:bound4k} in \cite{WildeFischer2023}, the existence of one such endpoint turns out to be insufficient to prove Theorem \ref{thm:bound2k+3}. However, the following lemma will fix this problem as it allows us to fix two endpoints.

\begin{lemma}\label{lem:leafchoice2}
Let $T$ be a binary phylogenetic tree on taxon set $X$ with $\vert X \vert = n\geq 3$. Let $x,y\in X$ such that $x\neq y$ and $x,y$ do not form a cherry in $T$. Let $q= \max\{2, \left\lfloor\frac{n-1}{2}\right\rfloor\}$. Then, $T$ contains $q$ edge-disjoint leaf-to-leaf paths $P_1,\ldots,P_{q}$ such that $x$ is an endpoint of $P_1$ and $y$ is an endpoint of $P_2$.
\end{lemma}

\begin{proof}
    Let $x,y\in X$ such that $x\neq y$ and $x,y$ do not form a cherry in $T$, and let $P = x, \beta_1, \ldots, \beta_m, y$ be the unique $x$-$y$-path in $T$. Note that as $x$ and $y$ do not form a cherry in $T$, we must have $m\geq 2$. Let $e=\{\beta_i, \beta_{i+1}\}$ for some $i\in \{1, \ldots, m-1\}$. If we delete $e$, we get two rooted binary phylogenetic trees $T_A$ and $T_B$. We denote by $A$ the taxon set of $T_A$ and by $B$ the taxon set of $T_B$. Then, we know that $\vert A\vert, \vert B\vert \geq 2$. This is due to the fact that leaf-to-leaf paths in a tree only contain two leaves (namely their endpoints), but contain only inner nodes in-between. This implies that both $\beta_i$ and $\beta_{i+1}$ are inner nodes, which due to binarity implies that they both have degree 3. Thus, $\beta_i$ does have one more neighbor (other than $x$ and $\beta_{i+1}$), say $v_a$, and $\beta_{i+1}$ does have one more neighbor (other than $y$ and $\beta_{i}$), say $v_b$, and both $v_a$ and $v_b$ are either leaves or can be regarded as roots of two rooted subtrees of the given tree $T$, both of which again contain leaves. So in any case, both $A$ and $B$ must contain more than one leaf. Now, let $a = \vert A\vert$ and $b = \vert B \vert$. By Lemma \ref{lem:leafchoice}, we know that $T_A$ contains $p := \left\lfloor \frac{a}{2} \right\rfloor$ leaf-to-leaf-paths $P_1, \ldots, P_p$ such that $x$ is an endpoint of $P_1$. Similarly, we know that $T_B$ contains $p':= \left \lfloor \frac{b}{2}\right \rfloor$ leaf-to-leaf-paths $Q_1, \ldots, Q_{p'}$ such that $y$ is an endpoint of $Q_1$. If we take these collections together we get $p + p'$ leaf-to-leaf-paths $P_1, \ldots, P_p, Q_1, \ldots, Q_{p'}$ such that $x$ is an endpoint in one of them and $y$ is an endpoint in another one. Thus, after re-naming the paths to make $y$ an endpoint of $P_2$ (i.e., we can merely swap paths $P_2$ and $Q_1$), it only remains to show $p + p' \geq q := \max\{2, \left\lfloor\frac{n-1}{2}\right\rfloor\}$. As $p \geq 1$ and $p' \geq 1$, it is clear that $p +p' \geq 2$. Furthermore, we have $p = \left\lfloor \frac{a}{2} \right\rfloor \geq \frac{a-1}{2}$ and $p'= \left\lfloor \frac{b}{2} \right\rfloor \geq \frac{b-1}{2}$. This implies $p + p' \geq \frac{a-1}{2} + \frac{b-1}{2} = \frac{a+b - 2}{2} = \frac{n-2}{2}$, where the latter equality holds as $a+b=\vert A \vert + \vert B \vert = \vert X \vert =n$. As $\left\lfloor\frac{n-1}{2}\right\rfloor - \frac{n-2}{2} \in \{0, \frac{1}{2}\}$ and as $p+p'$ is an integer, $p + p' \geq \frac{n-2}{2}$ implies $p + p' \geq \left \lfloor \frac{n-1}{2} \right \rfloor$ as desired. This completes the proof.
\end{proof}

Lemma \ref{lem:leafchoice2} has an important consequence. In fact, as the following proposition shows, it leads to a partial characterization of pairs of trees with identical $A_k$-alignments in case $n\geq 2k+1$. 

\begin{proposition}\label{prop:cherry}
Let $k \in \mathbb{N}_{\geq 1}$. Let $T$ and $\widetilde{T}$ be two binary phylogenetic $X$-trees with $\lvert X \rvert = n \geq 2k+1$. If $A_k(T) = A_k(\widetilde{T})$, then two leaves $x,y \in X$, $x\neq y$, form a cherry in $T$ if and only if they form a cherry in $\widetilde{T}$.
\end{proposition}

\begin{proof} Seeking a contradiction, we assume that the statement does not hold, i.e., we assume that there are two leaves $x,y\in X, x\neq y$, which form a cherry in $T$ but not in $\widetilde{T}$ even though $T$ and $\widetilde{T}$ have the same $A_k$-alignment. By Lemma \ref{lem:leafchoice2}, we know that with $q= \max\{2, \left\lfloor\frac{n-1}{2}\right\rfloor\}$, we can find edge-disjoint leaf-to-leaf paths $P_1, \ldots, P_q$ in $\widetilde{T}$ such that $x$ is an endpoint of $P_1$ and $y$ is an endpoint of $P_2$. As the case $k=1$ is clear from Proposition \ref{prop:A1A2} we can from now on assume $k\geq 2$. Using $n\geq 2k+1$, we get $q \geq  \left\lfloor\frac{n-1}{2}\right\rfloor \geq \left\lfloor\frac{(2k+1)-1}{2}\right\rfloor = k \geq 2$. This implies the existence of $k$ edge-disjoint leaf-to-leaf paths paths $P_1, \ldots, P_k$ such that $x$ is an endpoint of $P_1$ and $y$ is an endpoint of $P_2$. The idea of the proof is now to use this collection of paths for the construction of a character $f\in A_k(\widetilde{T}) \setminus A_k(T)$. This will give the wanted contradiction, as we have  $A_k(T) = A_k(\widetilde{T})$ by assumption.

    So for each $P_i$ with $i\in \{1, \ldots, k\}$, denote by  $a_i$ and  $b_i$ the endpoints of $P_i$. By the above explanations, we may assume that $a_1 = x$ and $a_2 = y$. 
    Now for every $v\in X$, we define $f(v)$ in the following way:
    \begin{equation*} 
    f(v)=\begin{cases} a & \text{if $v = a_i$ for some $i\in \{1, \ldots, k\}$},\\ b  &\text{else.}\end{cases}
    \end{equation*}

    We now denote by  $A_f=\{v \in X: f(v)=a\}$ and $B_f=\{v \in X: f(v)=v\}$ the sets of leaves that are assigned $a$ or $b$, respectively, by character $f$. We then show $f\in A_k(\widetilde{T})$. Using Proposition \ref{prop:menger0}, the existence of paths $P_1, \ldots, P_k$, which are $A_f$-$B_f$ leaf-to-leaf paths in $\widetilde{T}$, implies $l(f, \widetilde{T}) \geq k$. On the other hand, as $\vert A_f \vert =k$, it is clear that $l(f,\widetilde{T})\leq k$ (as all inner nodes could be assigned state $b$ by an extension and then in the worst case, each leaf in state $a$ would require one change). Thus, we have $l(f,\widetilde{T})= k$ and thus $f \in A_k(\widetilde{T})$ as desired.

    However, it remains to show $f\not\in A_k(T)$. Seeking a contradiction, let us assume this does not hold, i.e., let us assume that $f\in A_k(T)$. As $[x,y]$ is a cherry of $T$ and as $f(x) = f(y) = a$, we can apply Lemma \ref{lem:cherryred1} and get $f^1\in A_k(T^1)$, which shows that $l(f^1,T^1)=k$. 
    However, $\lvert A_{f^1} \rvert = k-1$, as $x,y\in A$ and as exactly one of the vertices $x$, $y$ is not contained in the taxon set of $T^1$. Thus, by the same argument as used above, we know $l(f^1,T^1)\leq k-1$, as no binary character with $k-1$ $a$'s can ever require more than $k-1$ changes on any tree. This contradiction shows that we must indeed have $f \not\in A_k(T)$.
    
    Thus, in summary, we indeed have found a character $f \in A_k(\widetilde{T})\setminus A_k(T)$, which shows $A_k(T)\neq A_k(\widetilde{T})$. This contradicts the assumption of the proposition and thus completes the proof.
    \end{proof}

\begin{remark}\label{rem2k} Before we continue, we point out that for each $k\geq 4$, there exists a pair of binary phylogenetic $X$-trees $T$ and $\widetilde{T}$ with $|X| = 2k$ and $A_k(T) = A_k(\widetilde{T})$ such that there is no cherry contained both in $T$ and in $\widetilde{T}$. A possible construction for this scenario is depicted in Figure \ref{fig:Counterexample}. Here, with $n=2k$, Lemma \ref{lem:leafchoice} guarantees the existence of $k$ edge-disjoint leaf-to-leaf paths. A choice (in this case, the only one) of such paths is  highlighted by bold edges in the figure. These $k$ paths result in $2^k$ binary characters using Proposition \ref{prop:menger0} (by choosing $a$ for one endpoint and $b$ for the other endpoint of each of these paths). Moreover, by Theorem \ref{thm:lengthAk}, we know for $n=2k$ that $\vert A_k(T)\vert = \vert A_k(\widetilde{T})\vert=2^k$. Thus, the $2^k$ characters constructed using the $k$ paths are the only elements both of $A_k(T)$ as well as of  $A_k(\widetilde{T})$. This shows that these sets indeed coincide, even though the trees have no cherry in common. Thus, the lower bound for $n$ in Proposition \ref{prop:cherry} is best possible.
\end{remark}

\begin{figure}[ht]
    \centering
\includegraphics{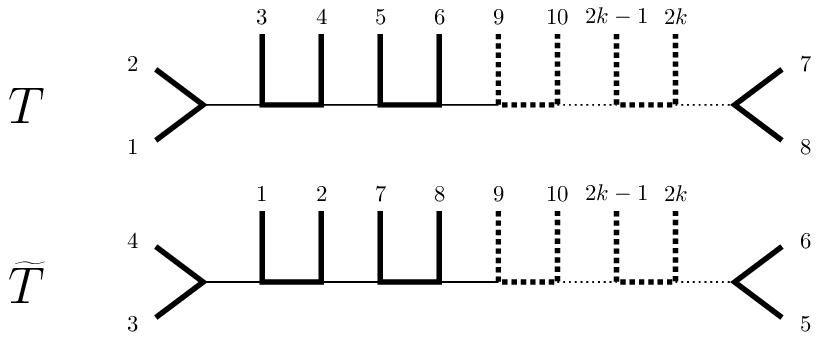}
\caption{A pair of binary phylogenetic $X$-trees with $|X|=2k$ and $k\geq 4$ and disjoint cherry sets. The dotted edges may or may not be there (depending on the chosen value of $k$). 
As described in Remark \ref{rem2k}, we have $A_k(T) = A_k(\widetilde{T})$, which can be shown using the edge-disjoint leaf-to-leaf paths highlighted in bold.  }
\label{fig:Counterexample}
\end{figure}

Proposition \ref{prop:cherry} together with Corollary \ref{cor:cherryred} suggests the following approach for the proof of Theorem \ref{thm:bound2k+3}: Given a pair of binary phylogenetic $X$-trees $T$ and $\widetilde{T}$ with $A_k(T) = A_k(\widetilde{T})$ and a value of $k\in \mathbb{N}_{\geq 1}$ as well as $n\geq 2k+3\geq 5$, by Proposition \ref{prop:cherry}, we know that there are at least two cherries which are contained in both trees. This is due to the fact that each phylogenetic tree with at least four leaves contains at least two disjoint cherries, and with $A_k(T) = A_k(\widetilde{T})$, both trees have all their cherries in common. Thus, the high-level idea of the proof of Theorem \ref{thm:bound2k+3} will use an inductive argument based on cherry reductions as follows: By Corollary \ref{cor:cherryred} and the inductive hypothesis, we will conclude that $T$ and $\widetilde{T}$ contain certain isomorphic subtrees with a large overlap. It then only remains to be shown that this overlap indeed implies  $T\cong \widetilde{T}$. However, for this last crucial step in the proof of Theorem \ref{thm:bound2k+3}, we require the following technical lemma.

\begin{lemma}\label{lem:preparation_for_proof}
Let $T$ be a binary phylogenetic $X$-tree with $|X| \geq 4$ and containing a cherry $[v,w]$. Denote the edges incident to $v$ and $w$ with $e_v$ and $e_w$, respectively. Furthermore, let $z'$ be the common neighbor of $v$ and $w$. Denote by $z$ the unique non-leaf vertex adjacent to $z'$ and by $E_z$ the set of edges incident with $z$. For $i=1,2$, let $T_i$ be a binary phylogenetic tree derived from $T$ by attaching cherry $[x,y]$ (with $x,y \not\in X$) to some edge $e_i \in E(T)\setminus\{e_v,e_w\}$. In particular, $T_i$ contains the cherry $[v,w]$ for $i=1,2$. Let $T_i^2$ denote the binary phylogenetic tree derived from $T_i$ by performing a cherry reduction of type $2$ using cherry $[v,w]$ for $i=1,2$. Then, we have: 

$$T_1^2 \cong T_2^2 \Longleftrightarrow (e_1=e_2 \mbox{ or }e_1, e_2 \in E_z).$$ 
\end{lemma}

\begin{proof}
Let $T$ with all of its vertices and edges as well as $T_1$, $T_2$, $T_1^2$ and $T_2^2$ be as described in the statement of the lemma. $T$ is schematically depicted in Figure \ref{fig:lem4T}. We now divide the proof into three parts.

\begin{figure}
\begin{center}
\includegraphics{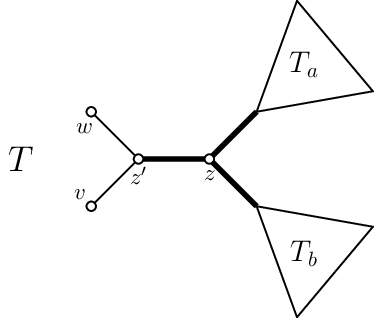}
\end{center}
\caption{Schematic depiction of tree $T$ containing cherry $[v,w]$ as described in Lemma \ref{lem:preparation_for_proof}. Here, the triangles $T_a$ and $T_b$ symbolize subtrees of $T$, each of which contains at least one leaf (as $\vert X \vert \geq 4$). The edges of $E_z$ are highlighted in bold.} \label{fig:lem4T}
\end{figure}

\begin{enumerate}
\item Assume $e_1=e_2$. We need to show that then $T_1^2 \cong T_2^2$. However, if $e_1=e_2$, by construction of $T_1$ and $T_2$, we  have $T_1 \cong T_2$, which immediately implies $T_1^2 \cong T_2^2$. So there remains nothing to show.
\item Now assume both $e_1$ and $e_2$ are contained in $E_z$. Again, we need to show that then $T_1^2 \cong T_2^2$. However, as $\vert E_z \vert = 3$ (as $T$ is binary), there are only three possibilities to attach cherry $[x,y]$ to edges in $E_z$, cf. Figure \ref{fig:lem4Txy}. So $T_1$ and $T_2$ are both contained in the set of three trees depicted in said figure. 

\begin{figure}
\begin{center}
\includegraphics{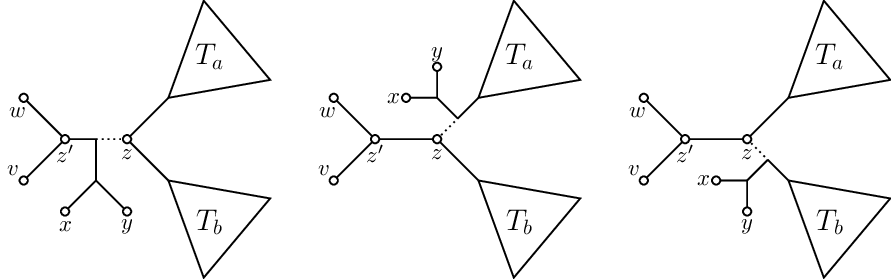}

\end{center}
\caption{Schematic depiction of the three possible trees resulting from attaching cherry $[x,y]$ to tree $T$ of Figure \ref{fig:lem4T}. Note that these trees are all NNI neighbors, which can be seen by considering subtrees swaps around the dotted edges.\label{fig:lem4Txy}}
\end{figure}

However, performing a cherry reduction of type 2 using cherry $[v,w]$ to any of these three trees results in the unique tree depicted in Figure \ref{fig:lem4T2}. This shows that $T_1^2 \cong T_2^2$ as desired.

\begin{figure}
\begin{center}
\includegraphics{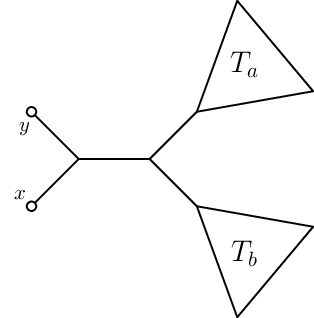}
\end{center}
\caption{Schematic depiction of the tree resulting from applying a cherry reduction of type 2 using cherry $[v,w]$ to any of the three trees from Figure \ref{fig:lem4Txy}.} \label{fig:lem4T2}
\end{figure}

\item Last, we need to show that if we have $e_1\neq e_2$ and additionally $\{e_1,e_2\}\not\subseteq E_z$, then $T_1^2 \not\cong T_2^2$. Without loss of generality, we may assume $e_1\not\in E_z$. As by choice of $e_1$ we also know that $e_1$ is not incident with $v$ or $w$, we know that $e_1$ must be contained in one of the subtrees $T_a$ or $T_b$ (with taxon sets $X_a$, $X_b$, respectively)  from Figure \ref{fig:lem4T}. So without loss of generality, let us assume $e_1$ is an edge of $T_a$. This in particular implies that $T_a$ contains more than one leaf. 
 Moreover, the removal of $e_1$ from $T$ would subdivide the tree into two subtrees, one of which contains both $v$ and $w$ and one of which does not. We consider the latter tree (which is a subtree of $T_a$) and call it $T_a'$ and its taxon set $X_a'$. However, this implies that $T_1$ contains the splits $\sigma_1:=X_a' \cup \{x,y\} \ \vert \ X \setminus X_a'$ as well as $\sigma_1':=X_a'  \ \vert \ (X \cup \{x,y\}) \setminus X_a'$. As $v,w \in X \setminus X_a'$, we accordingly derive that $T_1^2$ contains the splits $\sigma_2:=X_a' \cup \{x,y\} \ \vert \ X \setminus (X_a'\cup \{v,w\})$ and $\sigma_2':=X_a'  \ \vert \ (X\cup \{x,y\}) \setminus (X_a'\cup \{v,w\})$. 

If $\Sigma(T_2^2)$ does not contain $\sigma_2$ or $\sigma_2'$, there remains nothing to show: then, by Theorem \ref{thm:buneman}, we have $T_1^2 \not\cong T_2^2$ as desired.

So now assume  $\Sigma(T_2^2)$ contains both $\sigma_2$ and $\sigma_2'$, too. We now investigate what happens when cherry $[v,w]$ is re-attached to $T_2^2$ in order to derive $T_2$. Let us start with $\sigma_2$. From $\sigma_2\in \Sigma(T_2^2)$ we conclude that $T_2$ must contain either $\sigma_1$ or $\widetilde{\sigma}_1:= (X_a' \cup \{x,y\} \cup \{v,w\}) \ \vert \ X \setminus (X_a'\cup \{v,w\})$. If the latter was the case, the deletion of cherry $[x,y]$ in order to derive $T$ from $T_2$ would induce the split $\widetilde{\sigma}_1':= (X_a' \cup \{v,w\}) \ \vert \ X \setminus (X_a'\cup \{v,w\})$. However, as $T_a'$ is a proper subtree of $T_a$ in $T$, we can find a taxon $s \in X_a\setminus X_a'$. As $v,w \not \in X_a$, we know $s \in X \setminus (X_a'\cup \{v,w\})$. We now use this fact to show that $X_a \vert X\setminus X_a$ and $\widetilde{\sigma}_1'$ are not compatible. We have $\emptyset \neq X_a' \subseteq (X_a'\cup \{v,w\})\cap X_a$, $\emptyset\neq X_b \subseteq (X\setminus(X_a'\cup \{v,w\})\cap (X\setminus X_a)$ as well as $v,w \in X_a' \subseteq (X_a'\cup \{v,w\})\cap (X\setminus X_a)$ and $s \in X_a \cap (X \setminus (X_a'\cup \{v,w\}))$. As none of these four intersections are empty, the splits $X_a \vert X\setminus X_a$ and $\widetilde{\sigma}_1'$ are indeed not compatible by definition, and by Theorem \ref{thm:buneman}, this implies that their corresponding edges cannot both be contained in $T$. However, as we know that $X_a \vert X\setminus X_a$ is a split induced by $T$ via the existence of subtree $T_a$, we cannot have $\widetilde{\sigma}_1' \in \Sigma(T)$. Therefore, we must have $\sigma_1 \in \Sigma(T_2)$. Using an analogous  argument, it can also be seen that  $\sigma_1'\in \Sigma(T_2)$. 

Thus, re-attaching cherry $[v,w]$ to $T_2^2$ in order to derive $T_2$ leads to $\sigma_1, \sigma_1' \in \Sigma(T_2)$. Note that $\sigma_1$ implies the existence of a subtree with leaf set $X_a' \cup \{x,y\}$ in $T_2$. Similarly, $\sigma_1'$ implies the existence of a subtree with leaf set $X_a'$ in $T_2$. Together, these subtrees show  that when $T_2$ was constructed from $T$ by attaching cherry $[x,y]$, said cherry must have been attached to the edge inducing the $X$-split $X_a' \ \vert \ X\setminus X_a'$. However, this edge is precisely $e_1$, which leads to $e_1=e_2$, a contradiction. Thus, $\{\sigma_1,\sigma_1'\}\not\subset\Sigma(T_2^2)$, but  $\{\sigma_1,\sigma_1'\}\subset\Sigma(T_1^2)$, which shows that $\Sigma(T_1^2)\neq \Sigma(T_2^2)$.  By Theorem \ref{thm:buneman}, this implies $T_1^2 \not\cong T_2^2$ as desired. This completes the proof.

\end{enumerate}
\end{proof}

\subsubsection{Proof of Theorem \ref{thm:bound2k+3}}

We are now finally in the position to prove Theorem \ref{thm:bound2k+3}.

\begin{proof}[Proof of Theorem \ref{thm:bound2k+3}]
We first note that if we have $T \cong \widetilde{T}$ for two binary phylogenetic $X$-trees, then obviously $A_k(T)=A_k(\widetilde{T})$. So for this direction, there is nothing to show. For the other direction, however, we use induction on $k$ and assume $n\geq 2k+3$. Note that due to Proposition \ref{prop:A1A2}, for $k=1$ and $k=2$ there remains nothing to show. So we can assume $k\geq 3$. Furthermore, by the inductive hypothesis we can assume that for all pairs of binary phylogenetic $X'$-trees $T'$ and $\widetilde{T}'$ with $\lvert X'\rvert \geq 2(k-1)+3$, the $A_{k-1}$-alignments are identical if and only if $T'$ and $\widetilde{T}'$ are isomorphic.

Now we consider a pair $T$,  $\widetilde{T}$ of phylogenetic $X$-trees with $\lvert X \rvert \geq 2k+3$ and identical $A_k$-alignments. Our goal is to show that $T$ and $\widetilde{T}$ are isomorphic. By Proposition \ref{prop:cherry}, we know that $T$ and $\widetilde{T}$ have identical cherries. We consider two disjoint cherries $[v,w]$ and $[x,y]$ of $T$ (which exist because $|X| \geq 4$ as $k\geq 3$). Then $[v,w]$ and $[x,y]$ are also cherries of $\widetilde{T}$. Now let $T^2$ and $\widetilde{T}^2$ be the trees resulting from a cherry reduction of type $2$ using $[x,y]$ applied to $T$ and $ \widetilde{T}$, respectively, and let $S^2$ and $\widetilde{S}^2$ be the trees resulting from a cherry reduction of type $2$ using $[v,w]$ applied to $T$ and $ \widetilde{T}$, respectively. As $A_k(T)=A_k(\widetilde{T})$ by assumption and using Corollary \ref{cor:cherryred}, we can conclude that $A_{k-1}(T^2) = A_{k-1}(\widetilde{T}^2)$ and $A_{k-1}(S^2) = A_{k-1}(\widetilde{S}^2)$. Using the inductive hypothesis, we can then conclude that $T^2\cong \widetilde{T}^2$ and $S^2 \cong \widetilde{S}^2$. 

So $T^2$ and $S^2$ are subtrees of $T$ as well as of $\widetilde{T}$. As $[v,w]$ and $[x,y]$ are disjoint, deleting $[x,y]$ from $T$ or $\widetilde{T}$  leaves $[v,w]$ intact, which implies that $[v,w]$ is a cherry of $T^2$, too. Let $z'$ be the common neighbor of $v$ and $w$ in $T^2$. As an inner node of a binary phylogenetic tree, $z'$ has degree $3$, so there is a third neighbor $z$ of $z'$ in $T^2$ with $z\not\in \{v,w\}$. Let $E_z$ be the set of edges in $E(T^2)$ which are incident with $z$. Note that $\vert E_z \vert =3$ as $z$ cannot be a leaf (as $\vert X \vert \geq 4$). Let $e_1\in E(T^2)$ be such that attaching cherry $[x,y]$ to $e_1$ results in $T$ and let $e_2\in E(T^2)$ be such that attaching cherry $[x,y]$ to $e_2$ results in $\widetilde{T}$.

Now we apply Lemma \ref{lem:preparation_for_proof} to tree $T^2$ (with $T$ corresponding to $T_1$ in the lemma and  $\widetilde{T}$ corresponding to $T_2$ as well as $S^2$ to $T_1^2$, and $\widetilde{S}^2$ to $T_2^2$). This lemma allows us to conclude from $S^2 \cong \widetilde{S}^2$ that $e_1 = e_2$ or $e_1, e_2\in E_z$. However, if $e_1=e_2$, then $T$ and $\widetilde{T}$ are obtained by attaching the cherry $[x,y]$ to the same edge of $T^2$, which obviously implies $T \cong \widetilde{T}$, which completes the proof in this case.

If, on the other hand, $e_1\neq e_2$, we can conclude from $S^2 \cong \widetilde{S}^2$ that $e_1, e_2\in E_z$. However, in this case, $e_1$ and $e_2$ are adjacent, which makes $T$ and $\widetilde{T}$ NNI neighbors, cf. Figure \ref{fig:lem4Txy}. In this case, we can use Corollary \ref{cor:bound2k+3_nni} to conclude that $A_k(T) = A_k(\widetilde{T})$ implies $T \cong \widetilde{T}$. This completes the proof.
\end{proof}

\subsection{Characterization of pairs of trees with identical \texorpdfstring{$A_k$}{Ak}-alignment in case \texorpdfstring{$2k \leq n\leq 2k+2$}{2k<=2k+2}} \label{sec:n_smaller}

In the previous section, we have shown in Theorem \ref{thm:bound2k+3} that in case $n\geq 2k+3$, a pair of binary phylogenetic $X$-trees with $\vert X \vert = n$ has identical $A_k$-alignments if and only if the two trees are isomorphic. We have also explained that based on pre-knowledge from the literature, this result is best possible in the sense that constructions for pairs of trees with $n\leq 2k+2$ are known such that these trees share the same $A_k$-alignment. This immediately leads to the question if such pairs of trees, i.e., trees with identical $A_k$-alignments, can be characterized. Deriving such a characterization is the main aim of the present section. Therefore, we now consider $n\leq 2k+2$.

 However, note that the case $n<2k$ is not very interesting as $A_k(T) = \emptyset$ for every binary phylogenetic $X$-tree with $\lvert X \rvert < 2k$. This can be easily seen using Proposition \ref{prop:menger0}, because if a tree has fewer than $2k$ leaves, there cannot be $k$ edge-disjoint leaf-to-leaf paths, which implies that such a tree does not allow for the existence of a binary character with parsimony score $k$. Thus, we indeed have $A_k(T) = \emptyset$ in all these cases, which implies that \emph{all} binary phylogenetic trees with $n<2k$ have identical $A_k$-alignments.

 Therefore, in the present section it only remains to analyze the cases $n=2k$, $n=2k+1$ and  $n=2k+2$ in order to complete the characterization of pairs of binary phylogenetic trees with identical $A_k$-alignments. 

For these cases, our characterization is motivated by Proposition \ref{prop:nni}. By this result, we know that two NNI neighbors $T$ and $\widetilde{T}$ have identical $A_k$-alignments if and only if the condition $s(T,\widetilde{T})<k-2$ is satisfied. In this case, we call an NNI move from $T$ to $\widetilde{T}$ a \emph{$k$-problematic move}, which we will sometimes simply refer to as a \emph{problematic move} whenever there is no ambiguity. Moreover, we say  that a binary phylogenetic $X$-tree $\widetilde{T}$ can be obtained from binary phylogenetic $X$-tree $T$ by a \emph{series of $k$-problematic moves} (or shortly \emph{series of problematic moves}) if either $T \cong \widetilde{T}$ or if $\widetilde{T}$ is isomorphic to a binary phylogenetic $X$-tree obtained from $T$ by an iterative application of NNI moves, each of which is a $k$-problematic move.

Note that we use the term \enquote{problematic} in this context because such a series of moves changes the tree without changing the $A_k$-alignment, thereby destroying the uniqueness of the $A_k$-alignment under consideration. Thus, by Proposition \ref{prop:nni}, the existence of such a series of moves is a sufficient condition for arbitrary pairs of binary phylogenetic $X$-trees to have identical $A_k$-alignments. This observation immediately raises the question if \emph{every} pair of binary phylogenetic trees with identical $A_k$-alignments is characterized by the existence of problematic moves. It is the main aim of this section to answer this question affirmatively, which is done by the following theorem.

\begin{theorem}\label{thm:problematicmoves}
    Let $k\in \mathbb{N}_{\geq 2}$ and let $T$, 
 $\widetilde{T}$ be binary phylogenetic $X$-trees. Then, $A_k(T) = A_k(\widetilde{T})$ if and only if $\widetilde{T}$ can be obtained from $T$ by a series of $k$-problematic moves.
\end{theorem}

Before we can formally prove this theorem, we still need to present some preliminary results. As explained above, these will be mainly concerned with the cases $n=2k$, $n=2k+1$ and $n=2k+2$. Regarding Theorem \ref{thm:problematicmoves}, this is justified, as in all other cases, there remains nothing to show: 
\begin{itemize}
\item If $n<2k$, then $A_k(T) = \emptyset$ for every binary phylogenetic $X$-tree with $\lvert X \rvert = n$ as explained above. In this case, \emph{all} NNI moves are problematic as none of them changes the $A_k$-alignment. Note that the space of binary phylogenetic $X$-trees is connected under NNI \cite[Proposition 2.6.1]{Semple2003}, which implies that every such tree  $\widetilde{T}$ can be  obtained from any such tree $T$ by a series of NNI moves. As all of them are problematic, this proves Theorem \ref{thm:problematicmoves} in case $n<2k$.
\item If, on the other hand, $n\geq 2k+3$, then Theorem \ref{thm:bound2k+3} together with Proposition \ref{prop:nni} implies that there are no problematic moves at all. So in this case, Theorem \ref{thm:bound2k+3} and Theorem \ref{thm:problematicmoves} have the same meaning, which means that in the light of the previous section, there remains nothing to show in this case.
\end{itemize}

Therefore, it subsequently will suffice to prove Theorem \ref{thm:problematicmoves} for the cases $n=2k, n=2k+1$ and $n=2k+2$. Before we can do this, however, we need to turn our attention to some technical results which will be needed for the proof later on.

\subsubsection{Preparation for the proof of Theorem \ref{thm:problematicmoves}}

The function $s(T, \widetilde{T})$, which was defined in Proposition \ref{prop:nni}, plays a decisive role concerning  problematic moves.   We will use this function as well as Proposition \ref{prop:nni} now to derive the following result,  which will make it  easier to decide whether a certain NNI move is a problematic move. 

\begin{corollary}
   \label{cor:cases}
Let $k \in \mathbb{N}_{\geq 2}$, let $a\in \{0,1,2\}$ and let $n=2k+a$. Let $T$ be a binary phylogenetic $X$-tree with $\vert X\vert=n$. Let $\sigma=A\vert B \in \Sigma^\ast(T)$ be an $X$-split corresponding to an inner edge of $T$ and inducing subtrees $T_{A_1}$ and $ T_{A_2}$, whose leaves are subsets of $A$, as well as $T_{B_1}$ and $ T_{B_2}$, whose leaves are subsets of $B$, cf. Figure \ref{fig:treedecomp}. Let $n_1 = \vert A_1\vert$, $n_2 = \vert A_2\vert$, $n_3 = \vert B_1\vert$ and $n_4 = \vert B_2\vert$. Moreover, let $\widetilde{T}$ be the NNI neighbor of  $T$ obtained by exchanging $T_{A_2}$ with $T_{B_2}$. Then, we have: $A_k(T) = A_k(\widetilde{T})$ if and only if at most $2-a$ of the numbers $n_1, n_2, n_3, n_4$ are odd.
\end{corollary}

\begin{proof}
    Let $s(T,\widetilde{T}) = \sum\limits_{i=1}^4 \left\lfloor \frac{n_i-1}{2}\right\rfloor$ be the function defined in Proposition \ref{prop:nni}. It will be convenient for our proof to establish another formula for $s$. For $i\in \{1, \ldots, 4\}$,  let $\alpha_i = \left(\frac{n_i-1}{2}\right) - \left\lfloor \frac{n_i-1}{2}\right\rfloor$. Then, using $n = n_1 + n_2 + n_3 + n_4$, we obtain:
    \begin{equation*}
        s(T, \widetilde{T}) = \sum\limits_{i=1}^4 \left\lfloor \frac{n_i-1}{2}\right\rfloor = \sum\limits_{i=1}^4 \left(\frac{n_i-1}{2}\right) - \sum_{i=1}^4 \alpha_i = \frac{n}{2} - 2 - \sum_{i=1}^4 \alpha_i.
    \end{equation*}
    Now let $a\in \{0,1,2\}$ and $n=2k+a$. Then, we obtain: 
    \begin{equation*}
        s(T, \widetilde{T}) = \frac{n}{2} - 2 - \sum_{i=1}^4 \alpha_i = \frac{2k+a}{2} - 2 - \sum_{i=1}^4 \alpha_i = \left(k - 2\right) + \frac{a}{2} - \sum_{i=1}^4 \alpha_i.
    \end{equation*}
    By Proposition \ref{prop:nni},  we know that $A_k(T) = A_k(\widetilde{T})$ if and only if $s(T, \widetilde{T}) < k-2$. With the new formula for $s$ and using that $\alpha_i\in \left\{0, \frac{1}{2}\right\}$ for all $ i\in \{1, \ldots, 4\}$,  we conclude:
     \begin{align*}
        A_k(T) = A_k(\widetilde{T}) &\Longleftrightarrow s(T, \widetilde{T}) < k-2\\ & \Longleftrightarrow
        \left(k - 2\right) + \frac{a}{2} - \sum_{i=1}^4 \alpha_i  < k-2\\& \Longleftrightarrow 
        \sum_{i=1}^4 \alpha_i > \frac{a}{2}\\ &
       \Longleftrightarrow 
       \lvert \{i \in \{1, \ldots, 4\}: \alpha_i > 0\} \rvert > a\\ &
       \Longleftrightarrow
        \lvert \{i \in \{1, \ldots, 4\}: 
       \frac{n_i-1}{2} \not\in \mathbb{Z}\} \rvert > a\\ &
       \Longleftrightarrow 
       \lvert \{i \in \{1, \ldots, 4\}: n_i\text{ even}\} \rvert > a
    \end{align*}

    We now distinguish three cases depending on the value of $a$.
   \begin{itemize} 
   \item If $a=0$ and $n=2k$, then $\lvert \{i \in \{1, \ldots, 4\}: n_i\text{ even}\} \rvert$ has to be even, because $n$ is even (and $n = n_1 + n_2 + n_3 + n_4$). So $\lvert \{i \in \{1, \ldots, 4\}: n_i\text{ even}\} \rvert > 0$ implies $\lvert \{i \in \{1, \ldots, 4\}: n_i\text{ even}\} \rvert > 1$. In this case, $A_k(T) = A_k(\widetilde{T})$ if and only if $\lvert \{i \in \{1, \ldots, 4\}: n_i\text{ even}\} \rvert \in \{2,4\}$. Using the fact that $n$ is even, the latter is the case if and only if at most $2=2-a$ of the numbers $n_i$ with $i\in \{1,\ldots,n\}$ are odd.

   \item If $a=1$ and $n=2k+1$, then $\lvert \{i \in \{1, \ldots, 4\}: n_i\text{ even}\} \rvert$ has to be odd, because $n$ is odd. So $\lvert \{i \in \{1, \ldots, 4\}: n_i\text{ even}\} \rvert > 1$ implies $\lvert \{i \in \{1, \ldots, 4\}: n_i\text{ even}\} \rvert = 3$. In this case, $A_k(T) = A_k(\widetilde{T})$ if and only if exactly one of the numbers $n_i$ (with $i\in \{1, \ldots, 4\}$) is odd. Using the fact that $n$ is odd, the latter holds if and only if at most $1=2-a$ of the numbers $n_i$ with $i\in \{1,\ldots,n\}$ are odd.

  \item  If $a=2$ and $n=2k+2$, then $\lvert \{i \in \{1, \ldots, 4\}: n_i\text{ even}\} \rvert$ has to be even, because $n$ is even. So $\lvert \{i \in \{1, \ldots, 4\}: n_i\text{ even}\} \rvert > 2$ implies $\lvert \{i \in \{1, \ldots, 4\}: n_i\text{ even}\} \rvert = 4$. In this case, $A_k(T) = A_k(\widetilde{T})$ if and only if all of the numbers $n_i$ (with $i\in \{1, \ldots, 4\}$) are even. This is equivalent to the requirement that at most $0=2-a$ of the numbers $n_i$ with $i\in \{1,\ldots,n\}$ are odd.
\end{itemize}
    This completes the proof.
\end{proof}

The main idea of the proof of Theorem \ref{thm:problematicmoves} will be as follows. Given two trees $T$ and $\widetilde{T}$ with $A_k(T) = A_k(\widetilde{T})$, we will first show that we can assume that $T$ and $\widetilde{T}$ have a common cherry, which will enable us to use Corollary \ref{cor:cherryred}. Second,  we will conclude inductively that $\widetilde{T}^2$ can be obtained from $T^2$ by a series of problematic moves and derive a similar statement about the pair $T$ and $\widetilde{T}$. For this purpose we need the following lemma.

\begin{lemma}\label{lem:series_problematic_moves}
    Let $k\in \mathbb{N}_{\geq 2}$ and let $S$ and $\widetilde{S}$ be binary phylogenetic $X$-trees with $\lvert X \rvert = n = 2k+a$ for some $a\in \{0,1,2\}$. Let $X^{\ast} = X\cup \{x,y\}$ such that $x\neq y$ and $x,y\not\in X$. Assume that $\widetilde{S}$ is obtained from $S$ by a series of $k$-problematic moves. Let $e_1\in E(S)$. Let $T$ be the binary phylogenetic $X^{\ast}$-tree obtained from $S$ by subdividing $e_1$ and attaching the cherry $[x,y]$ to $e_1$. 
    
    Then, there exists an edge $e_2\in E(\widetilde{S})$ for which the following statement is true: Let $\widetilde{T}$ be the binary phylogenetic $X^{\ast}$-tree obtained from $\widetilde{S}$ by attaching the cherry $[x,y]$ to $e_2$. Then, $\widetilde{T}$ can be obtained from $T$ by a series of $(k+1)$-problematic moves.
\end{lemma}

\begin{proof}
Let $d$ be the minimal number of $k$-problematic moves needed to obtain $\widetilde{S}$ from $S$. We use induction on $d$. For the base case $d=0$, for which we have $S \cong \widetilde{S}$, there is nothing to show as we can choose $e_1=e_2$.\footnote{Note that starting the induction with $d=0$ is possible as we included the case $S \cong \widetilde{S}$ in our definition of a series of $k$-problematic moves.} 

We now assume that the statement of the lemma already holds for all pairs of trees which are connected by at most $d-1$ $k$-problematic moves and consider a pair $S$, $\widetilde{S}$ such that $d>0$ $k$-problematic moves are needed to derive $\widetilde{S}$ from $S$. Let $S_2$ be some $X$-tree obtained from $S$ by a series of $d-1$ $k$-problematic moves such that $\widetilde{S}$ is obtained from $S_2$ by a single $k$-problematic move; i.e., $S_2$ is the last tree on a shortest path from $S$ to $\widetilde{S}$ before reaching $\widetilde{S}$. We can then apply the inductive  hypothesis to the pair $S$, $S_2$. Hence, we can assume that there exists some $f\in E(S_2)$ with the following property: Let $T_2$ be the $X^{\ast}$-tree obtained from $S_2$ by attaching cherry $[x,y]$ to $f$. Then, $T_2$ is obtained from $T$ by a series of $(k+1)$-problematic moves. So now the only thing that remains to be shown is the existence of some $e_2\in E(\widetilde{S})$ as described in the lemma. 

Let $e = \{v,w\}\in E(S_2)$ be the inner edge used by a $k$-problematic move applied to $S_2$ to obtain $\widetilde{S}$. Consider the $X$-split $A\vert B$ induced by $e$, and let $S_A$ and $S_B$ with taxon sets $A$ and $B$, respectively, be the two rooted binary subtrees resulting from $S_2$ when $e$ is removed. Without loss of generality, $v$ is the root of $S_A$ and $w$ is the root of $S_B$. Furthermore, let $S_{A_1}, S_{A_2}$ be the maximal pending subtrees of $S_A$ and $S_{B_1}, S_{B_2}$ be the maximal pending subtrees of $S_B$.  
Then there are edges $e_1 = \{v, v_1\}, e_2 = \{v,v_2\}, e_3 = \{w,w_1\}, e_4 = \{w, w_2\}$ such that $v_1$ is the root of $S_{A_1}$, $v_2$ is the root of $S_{A_{2}}$, $w_1$ is the root of $S_{B_1}$ and $w_2$ is the root of $S_{B_2}$. Without loss of generality, we may assume that the $k$-problematic move applied to $S_2$ to obtain $\widetilde{S}$ deletes $e_1$ and $e_3$ and replaces them with $f_1 = \{w,v_1\}$ and $f_3 = \{v, w_1\}$.

Now we show how to obtain a binary phylogenetic $X^{\ast}$-tree $\widetilde{T}$ as described in the lemma by a series of $(k+1)$-problematic moves from $T_2$. We distinguish the case in which the edges $e$ (the one used for the $k$-problematic move from $S_2$ to $\widetilde{S}$) and $f$ (the one to which the cherry $[x,y]$ gets attached to derive $T_2$ from $S_2$) are different from the case in which they are the same.

\emph{Case $e \neq f$:} In this case, we have either $f\in E(S_A)$ or $f\in E(S_B)$. Without loss of generality, we may assume $f\in E(S_A)$. As attaching $[x,y]$ to $f$ in order to form $T_2$ does not affect $e$, we have $e\in E(T_2)$. Now deleting $e$ from $T_2$ leads to two rooted phylogenetic trees $T_{A^{\ast}}$ and $T_B$ with taxon sets $A^{\ast} = A\cup \{x,y\}$ and $B$, respectively. $T_{A^{\ast}}$ has two maximal pending  subtrees $T_{A_1^{\ast}}$ and $T_{A_2^{\ast}}$ with taxon sets $A_1^{\ast}$ and $A_2^{\ast}$. Then either $A_1^{\ast} = A_1$ and $A_2^{\ast} = A_2\cup \{x,y\}$ or $A_1^{\ast} = A_1\cup \{x,y\}$ and $A_2^{\ast} = A_2$. Similarly, $T_B$ has maximal pending subtrees $T_{B_1}$ and $T_{B_2}$, which actually coincide with $S_{B_1}$ and $
S_{B_2}$, respectively. 

Consider the case $A_1^{\ast} = A_1$ and $A_2^{\ast} = A_2\cup \{x,y\}$.   Let $n_1 = \vert A_1\vert = \vert A^{\ast}_1\vert$, $n_2 = \vert A_2\vert$, $n_3 = \vert B_1\vert$, $n_4 = \vert B_2\vert$,  and $n^{\ast}_2 = \vert A^{\ast}_2\vert = \vert A_2\vert + 2$. In the following,  let $\widetilde{T}$ be the result of the NNI move operating on $T_2$ defined by $e$ and by switching $T_{A_1^\ast}$ with $T_{B_1}$. We want to show that this NNI move is a $(k+1)$-problematic move (because this move, together with the $(k+1)$-problematic moves from $T$ to $T_2$, will form the desired series of $(k+1)$-problematic moves from $T$ to $\widetilde{T}$). 

Now, as $e$ is defining a $k$-problematic move from $S_2$ to $\widetilde{S}$ and as we have $n=2k+a$ by assumption, by Corollary \ref{cor:cases} we know that at most $2-a$ elements of the numbers $n_1, n_2, n_3, n_4$ are odd. As adding $2$ to one of them does not change the parity, this immediately implies that at most $2-a$ elements of the numbers $n_1, n^{\ast}_2, n_3, n_4$ are odd. But this is exactly the condition of Corollary \ref{cor:cases} for $\vert X^{\ast} \vert = n+2 = 2(k+1)+a$. So the NNI move from $T_2$ to $\widetilde{T}$ is indeed a $(k+1)$-problematic move. Using the assumption that $T_2$ is obtained from $T$ by a series of problematic moves, we conclude that $\widetilde{T}$ is obtained from $T$ by a series of problematic moves. The same argument works if we consider the case $A_1^{\ast} = A_1\cup \{x,y\}, A_2^{\ast} = A_2$ and replace the quadruple $n_1, n^{\ast}_2, n_3, n_4$ with $n^{\ast}_1, n_2, n_3, n_4$ and $n^{\ast}_1 = n_1 + 2$ in the course of the argument.

It remains to show that $\widetilde{T}$ is obtained from $\widetilde{S}$ and some $e'\in E(\widetilde{S})$ by attaching the cherry $[x,y]$ to $e'$. We know that $T_2$ is obtained from $S_2$ and $f$ in this way. Now let us assume again at first $A_1^{\ast} = A_1$ and $A_2^{\ast} = A_2\cup \{x,y\}$. This assumption implies that $f\in E(S_{A_2}) \cup \{e_2\}$ and that $T_{A^{\ast}_2}$ is obtained from $S_{A_2}$ by attaching cherry $[x,y]$ to $f$. Remember that the previously defined NNI move from $S_2$ to $\widetilde{S}$ leaves all edges unchanged with the exception of $e_1$ and $e_3$. But $e_1, e_3\not\in E(S_{A_2}) \cup \{e_2\}$ by definition of $S_{A_2}$, and hence $E(S_{A_2}) \cup \{e_2\} \subset E(\widetilde{S})$. Thus, the NNI move from $S_2$ to $\widetilde{S}$ leaves the subtree with edge set $E(S_{A_2}) \cup \{e_2\}$ unchanged. Exactly the same argumentation implies that the NNI move from $T_2$ to $\widetilde{T}$ (which was defined by $e$ and leaves all edges unchanged with the exception of $e_1$ and $e_3$)  leaves the subtree with edge set $E(T_{A^{\ast}_2}) \cup \{e_2\}$ unchanged. As $T_2$ was obtained from $S_2$ by attaching $[x,y]$ to $f\in E(S_{A_2}) \cup \{e_2\}$, we can now set $e' = f$ and see that $\widetilde{T}$ is obtained from $\widetilde{S}$ by attaching $[x,y]$ to $e'$. 

The argumentation for the case $A_1^{\ast} = A_1\cup \{x,y\}, A_2^{\ast} = A_2$ contains an additional subtlety due to the possibility that we might have $f=e_1$, in which case $f$ would be deleted by the NNI operation leading from $S_2$ to $\widetilde{S}$. However, let us first consider the case $f\neq e_1$. Then, $f\in E(S_{A_1})$, and the NNI move from $S_2$ to $\widetilde{S}$ leaves $S_{A_1}$ unchanged. As before, one can argue that the NNI move from $T_2$ to $\widetilde{T}$ (which is defined by $e$ and leaves all edges unchanged with the exception of $e_1$ and $e_3$)  leaves the rooted subtree $T_{A^{\ast}_1}$ unchanged. In this case, again we can set $e' = f$ and see that $\widetilde{T}$ is obtained from $\widetilde{S}$ by attaching $[x,y]$ to $e'$. 

Now we consider the case $f=e_1$, which is depicted in Figure \ref{fig:LemmaProblematicMovesA}. In this case, let $u$ be the vertex resulting from subdividing $f$ such that $u$ is adjacent to the common neighbor of $x$ and $y$ in $T_2$. Then $u$ is root of $T_{A^{\ast}_1}$, and the NNI move from $T_2$ to $\widetilde{T}$ can be described as leaving all edges unchanged with the exception of $e^{\ast}_1 = \{u, v\}$ and $e_3$ which are replaced with $f^{\ast}_1 = \{u, w\}$ and $f_3=\{v,w_1\}$. As $T_2$ is obtained from $S_2$ by attaching $[x,y]$ to $e_1$,  we have that $\widetilde{T}$ is obtained from $\widetilde{S}$ by attaching $[x,y]$ to $f_1$. But this implies that we can set $e' = f_1$ and see that $\widetilde{T}$ is obtained from $\widetilde{S}$ by attaching $[x,y]$ to $e'$. 

\emph{Case $e=f$:} We will show that this case can be reduced to the previous case. Let us start with subdividing edge $f=e=\{v,w\}$ with a degree-2 vertex $u$, which results in two new edges $f_1 = \{v,u\}$ and $f_2 = \{u,w\}$. After attaching cherry $[x,y]$ to $u$, the common neighbor $z$ of $x$ and $y$ is also adjacent to $u$. Deleting $f_1$ leads to two rooted subtrees $T_A$ and $T_{B^{\ast}}$ such that $T_A$ has taxon set $A$ and root $v$ and $T_{B^{\ast}}$ has taxon set $B^{\ast} = B \cup \{x,y\}$ and root $u$. Deleting $u,v$ and the edges incident with $u$ or $v$ leads to the rooted subtrees $T_{A_1}, T_{A_2}, T_B$ and $T_{x,y}$ such that $T_{A_1}$ has taxon set $A_1$, $T_{A_2}$ has taxon set $A_2$, $T_B$ has taxon set $B$ and $T_{x,y}$ has taxon set $\{x,y\}$ and root $z$. Using Corollary \ref{cor:cases}, we will now show that $f_1$ defines a $(k+1)$-problematic move. 
In order to do so, consider the quadruple $n_1 = \vert A_1 \vert, n_2 = \vert A_2 \vert, n_3+n_4 = \vert B\vert, 2 = \vert\{x,y\}\vert$ consisting of the sizes of the taxon sets of $T_{A_1}, T_{A_2}, T_B$ and $T_{x,y}$, respectively. 

Using the assumption $n = 2k+a$ for some $a\in \{0,1,2\}$, we conclude $\vert X^{\ast} \vert = n+2 = 2(k+1) + a$. By Corollary \ref{cor:cases}, we see that $A_k(S_2) = A_k(\widetilde{S})$ implies that at most $2-a$ elements of the quadruple $n_1, n_2, n_3, n_4$ are odd. But then by a simple parity argument, at most $2-a$ elements of the triple $n_1, n_2, n_3+n_4$ are odd. This implies that at most $2-a$ elements of the quadruple $n_1, n_2, n_3+n_4, 2$ are odd, which is exactly the condition of Corollary \ref{cor:cases}. Thus, $f_1$ defines only $(k+1)$-problematic moves.

One of the two possible NNI moves defined by $f_1$ can be characterized as replacing the edges $\{v,v_1\}, \{u,z\}\in E(T_2)$ with the edges $\{v,z\}, \{u,v_1\}$. Let $T'_2$ be the binary phylogenetic $X^{\ast}$-tree obtained in this way (cf. Figure \ref{fig:LemmaProblematicMovesB}). As $T_2$ is obtained from $T$ by a series of $(k+1)$-problematic moves and as $T'_2$ is obtained from $T_2$ by a $(k+1)$-problematic move, clearly $T'_2$ is obtained from $T$ by a series of $(k+1)$-problematic moves. But now there is another way to obtain $T'_2$: We could have first subdivided the edge $e_2=\{v, v_2\} \in E(S_2)$ and then attached $[x,y]$ to this edge. So now if we consider $T$ and $T_2'$ instead of $T$ and $T_2$, this implies that we have reduced the case $e=f$ to a case where the NNI edge and the attachment edge for the cherry are different, i.e., $e \neq f$. As this case was already proven, this completes the proof.
\end{proof}

\begin{figure}
    \centering
\includegraphics{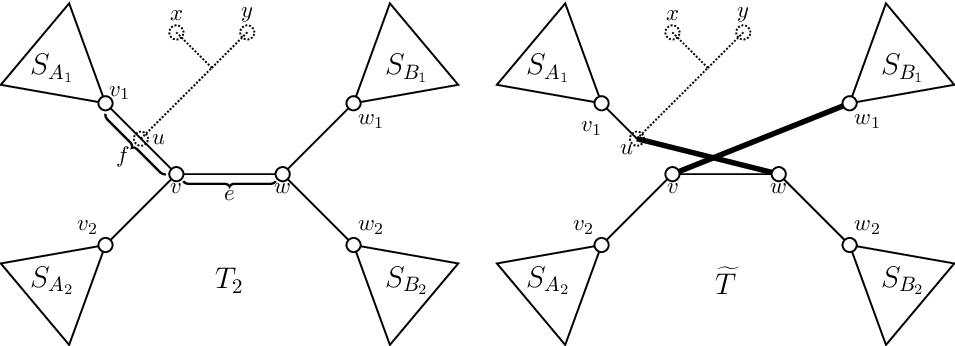}
\caption{One of the subcases of the case $e\neq f$ in the proof of Lemma \ref{lem:series_problematic_moves}. Here, $T_2$ is transformed to $\widetilde{T}$ by replacing the edges $\{u, v\}$ and $\{w,w_1\}$ with $\{u, w\}$ and $\{v, w_1\}$ (thick lines). If we delete the edges indicated by dotted lines and suppress the resulting degree-$2$ vertex $u$, we obtain $S_2$ (left) as well as $\widetilde{S}$ (right).}
\label{fig:LemmaProblematicMovesA}
\end{figure}

\begin{figure}
\includegraphics{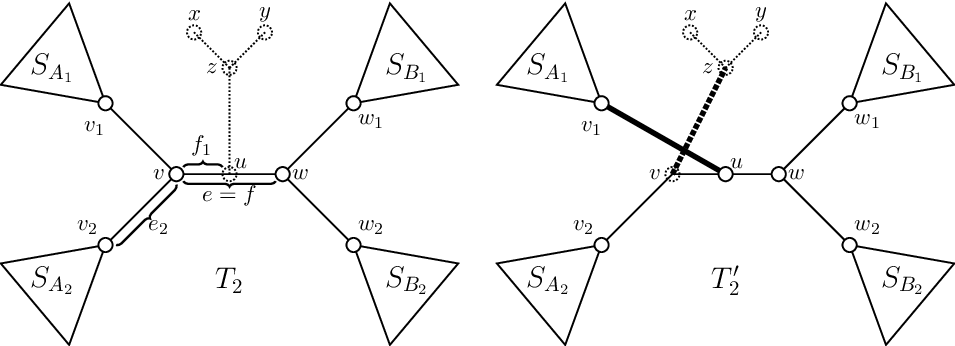}
\caption{Case $e = f$ in the proof of Lemma \ref{lem:series_problematic_moves}. Here,  $T_2$ is turned into $T'_2$ by replacing the edges $\{u, z\}$ and $\{v,v_1\}$ with $\{u, v_1\}$ and $\{v, z\}$ (thick edges). If we delete the edges indicated by dotted lines and suppress the resulting degree-$2$ vertex ($u$ in the left tree, $v$ in the right tree), then we obtain $S_2$ in both cases. However, cherry $[x,y]$ is attached to different edges of $S_2$ in each case ($\{v,w\}$ on the left and $\{u,v_2\}$ on the right). } 
\label{fig:LemmaProblematicMovesB}
\end{figure}

Recall that Proposition \ref{prop:cherry} stated that in case $n \geq 2k+1$, if $A_k(T) = A_k(\widetilde{T})$ for a pair of binary phylogenetic $X$-trees $T$ and $\widetilde{T}$, then they have all their cherries in common. We have already seen in Remark \ref{rem2k} that the same does not necessarily hold if $n=2k$. Therefore, we will require different tools to tackle the case $n=2k$, which will be introduced by the following subsection.

\paragraph{The case \texorpdfstring{$n=2k$}{n=2k}}

\par\vspace{0.5cm}\mbox{}

We start with the following slightly technical lemma, which links $A_k$-alignments to leaf-to-leaf paths.

\begin{lemma}
    \label{lem:alignment_n=2k}
    Let $k\in \mathbb{N}_{\geq 1}$ and let $T$ be a binary phylogenetic $X$-tree with $\vert X\vert = n = 2k$. Let $P_1, \ldots, P_k$ be edge-disjoint leaf-to-leaf paths (which exist due to Lemma \ref{lem:leafchoice}) such that each $P_i$ has endpoints $a_i, b_i$ (for $i\in \{1, \ldots, k\}$). For each $p\in \{0,1\}^k$, we then construct the following binary character $f_p: X \rightarrow \{a,b\}$:
    \begin{equation*} 
    f_p(v)=\begin{cases} a & \text{if ($v = a_i$ and $p_i = 0$) or ($v= b_i$ and $p_i = 1$) for some $i\in \{1, \ldots, k\}$ },\\ b  &\text{else.}\end{cases}
    \end{equation*}
    Then, $p\mapsto f_p$ gives us a bijection between $\{0,1\}^k$ and $A_k(T)$.
\end{lemma}

\begin{proof}
    We first show that $f_p\in A_k(T)$ for all $p\in \{0,1\}^k$. From the definition of $f_p$, we can immediately conclude that $f_p(a_i) \neq f_p(b_i)$ for each $i\in \{1,\ldots, k\}$. Also, $P_1, \ldots, P_k$ were chosen to be edge-disjoint leaf-to-leaf-paths, which by the definition of $f_p$ now also are edge-disjoint $A_{f_p}$-$B_{f_p}$-paths. Thus,  Proposition \ref{prop:menger0} implies $l(f_p, T) = k$ and thus $f_p \in A_k(T)$. 

    The map $p\mapsto f_p$ is clearly injective. To see this, let $p, q\in \{0,1\}^k$ be such that $p\neq q$, so there must be some $i\in \{1, \ldots, k\}$ for which $p$ and $q$ differ. Without loss of generality, assume  $p_i = 0$ and $ q_i = 1$. Then, by definition of $f_p$ and $f_q$, we conclude $f_p(a_i) = a \neq b = f_q(a_i)$. Thus,  $f_p \neq f_q$, which shows injectivity. 

However, this implies that the $2^k$ possible choices of $p$ induce $2^k$ different choices of elements of $A_k(T)$. Using Theorem \ref{thm:lengthAk} with $n=2k$, we know that $\vert A_k(T) \vert = 2^k$. So in fact, \emph{all} elements of $A_k(T)$ are covered by the different choices of $p$, which shows surjectivity and thus completes the proof.
\end{proof}

Next, we show that the choice of paths in Lemma \ref{lem:alignment_n=2k}, whose existence is guaranteed by Lemma \ref{lem:leafchoice}, is actually unique if $n=2k$.

\begin{lemma}\label{lem:2kuniquepaths} Let $k\in \mathbb{N}_{\geq 1}$ and let $T$ be a binary phylogenetic $X$-tree with $\vert X\vert = n = 2k$. Let $\mathcal{P}=\{P_1, \ldots, P_k\}$ as well as $\mathcal{Q}=\{Q_1,\ldots, Q_k\}$ be sets of edge-disjoint leaf-to-leaf paths in $T$. Then, we have $\mathcal{P}=\mathcal{Q}$. In other words, there is only one set of $k$ edge-disjoint leaf-to-leaf paths in $T$.
\end{lemma}

\begin{proof} We prove the statement by induction on $k$. For $k=1$, $T$ consists of only one edge, which clearly shows that there is precisely one leaf-to-leaf path, so there is nothing to show. So let us assume the statement is true for $k$ and let us assume we now have a binary phylogenetic tree $T$ with $n=2(k+1)=2k+2\geq 4$ leaves. Then, $T$ has a cherry $[x,y]$, and we perform a cherry reduction of type 2 using this cherry in order to derive tree $T'$. $T'$ has $2k+2-2=2k$ leaves, so by the inductive hypothesis, there is precisely one set of $k$ edge-disjoint leaf-to-leaf paths $\mathcal{P}'=\{P_1,\ldots,P_k\}$ in $T'$. Clearly, each of these paths corresponds to a path in $T$. We now add path $P_{k+1}=x,z,y$ to $\mathcal{P}'$ in order to derive a set $\mathcal{P}$ of $k+1$ edge-disjoint leaf-to-leaf paths in $T$, where $z$ is the unique vertex adjacent to both $x$ and $y$. So the only remaining thing to show is uniqueness. However, this is clear as \emph{every} collection of $k+1$ edge-disjoint leaf to leaf paths in $T$ with $n=2(k+1)$ needs to cover \emph{all} leaves of $T$, so every such collection needs to contain path $P_{k+1}=x,z,y$. This is due to the fact that if $x$ or $y$ were endpoints of any other path $P'$ in such a collection, the respective other leaf of the cherry could not be reached anymore by any other leaf-to-leaf path $P''$ edge-disjoint to $P'$ (as both $P'$ and $P''$ would necessarily contain the edge $e=\{z,z'\}$, where $z'$ is the unique vertex with the properties that $z'$ is adjacent to $z$ and $z\neq x,y$). Thus, the way $\mathcal{P}$ was constructed is the only way to derive $k+1$ edge-disjoint leaf-to-leaf paths in $T$, which completes the proof. \end{proof}

The following proposition gives a first characterization of pairs of trees with identical $A_k$-alignments for the case $n=2k$.

\begin{proposition}\label{prop:characterization2k} Let $T$ and $T'$ be two binary phylogenetic $X$-trees with $\vert X \vert = 2k$, where $k \in \mathbb{N}_{\geq 1}$. Let $P_1,\ldots,P_k$ and $P_1',\ldots,P_k'$ be the unique collections of $k$ edge-disjoint leaf-to-leaf-paths for $T$ and $T'$, respectively, as derived by Lemma \ref{lem:2kuniquepaths}. Let $\mathcal{E}=\{\{ a_i,b_i\}: \mbox{ $a_i$ and $b_i$ are endpoints of $P_i$} \}$ and let $\mathcal{E}'=\{\{ a_i',b_i'\}: \mbox{ $a_i'$ and $b_i'$ are endpoints of $P_i'$} \}$. Then, we have $A_k(T)=A_k(T')$ if and only if $\mathcal{E}=\mathcal{E}'$.
\end{proposition}

\begin{proof} We first note that two leaves $x,y \in X$ are endpoints of one of the paths $P_i$ (or $P_i'$, respectively) for some $i \in \{1,\ldots,k\}$ if and only if they are assigned a different state by every $f\in A_k(T)$ (or $f \in A_k(T')$, respectively). This is due to Lemma \ref{lem:alignment_n=2k}, which shows that there is a bijection from the $\{0,1\}^k$ to $A_k(T)$ (or $A_k(T')$, respectively),  and this bijection is based on the (by Lemma \ref{lem:2kuniquepaths}: unique) paths $P_1,\ldots,P_k$ (or $P_1',\ldots,P_k'$), which assigns different states to each endpoint of each path. So it is quite easy to derive the endpoints of each such path from an $A_k$-alignment: For each pair of leaves $x,y \in X$ which differ for \emph{every} character in the alignment, there has to be a path connecting them. All other pairs of leaves are not connected by paths in the unique set of $k$ edge-disjoint leaf-to-leaf paths. This immediately shows that if $A_k(T)=A_k(T')$, the sets $\mathcal{E}$ and $\mathcal{E}'$, which contain the pairs of endpoints of said paths, must be equal, too.
Moreover, if we have $\mathcal{E}=\mathcal{E}'$, we know that the pairs of leaves that get connected by said paths coincide. By Lemma \ref{lem:alignment_n=2k}, the elements of $A_k(T)$ are precisely the $2^k$ choices of $f_p$, which are identical for $T$ and $T'$ if  $\mathcal{E}=\mathcal{E}'$, as the set containing all $f_p$ only depends on the endpoints of the path set under consideration. This shows that $A_k(T)=A_k(T')$ and thus completes the proof.
\end{proof}

We need two more technical corollaries resulting from Lemma \ref{lem:2kuniquepaths} before we can continue.  

\begin{corollary} \label{cor:even_subtrees}
    Let $k\in \mathbb{N}_{\geq 1}$, and let $T$ be a binary phylogenetic $X$-tree with $\vert X\vert = n = 2k$. Furthermore, let $\mathcal{P}$ be the set of $k$ edge-disjoint leaf-to-leaf paths in $T$ (which is unique by Lemma \ref{lem:2kuniquepaths}), and let $P\in \mathcal{P}$. Assume $P = a, \beta_1, \ldots, \beta_m , b$, where $a, b$ are endpoints of $P$ and the $\beta_i$'s are internal vertices.  Denote by $T_{A_1}, \ldots, T_{A_m}$  the rooted binary phylogenetic trees resulting from deleting the vertices of $P$ from $T$, such that each $T_{A_i}$ has taxon set $A_i$ for $i\in \{1, \ldots, m\}$. Let $n_i = \vert A_i \vert$ for $i\in \{1, \ldots, m\}$. Then, we have: all of the numbers $n_i$ (for ($i\in \{1, \ldots, m\}$) are even.
\end{corollary}

\begin{proof}
    Seeking a contradiction, we assume that there is a $j \in \{1,\ldots,m\}$ such that the number of leaves $n_j$ of $T_{A_j}$ (the subtree adjacent to $\beta_j$) is odd. 
    Then, at least one leaf $z$ of $T_{A_j}$ is the endpoint of a path $\widehat{P} \in \mathcal{P}$ whose other endpoint is \emph{not} contained in $T_j$ (as again all leaves are endpoints of one of these paths as there are $k$ paths and we have $n=2k$). However, then $P$ and $\widehat{P}$ cannot be edge-disjoint  as both paths share at least one edge incident with $\beta_j$. This is a contradiction, as $P$ and $\widehat{P}$ are both contained in $ \mathcal{P}$, which is a set of edge-disjoint paths. This shows that the assumption was wrong, which completes the proof of the corollary.
\end{proof}

\begin{corollary} \label{lem:cherry_n=2k}
    Let $k\in \mathbb{N}_{\geq 1}$ and let $T$, $\widetilde{T}$ be a pair of binary phylogenetic $X$-trees with $\vert X\vert = n = 2k$ and such that $A_k(T) = A_k(\widetilde{T})$. Let $[x,y]$ be a cherry in $T$ and let $P = x, \beta_1, \ldots, \beta_m , y$ be the unique $x-y$-path in $\widetilde{T}$. Denote by $T_{A_1}, \ldots, T_{A_m}$  the rooted binary phylogenetic trees resulting from deleting the vertices of $P$, such that each $T_{A_i}$ has taxon set $A_i$ for $i\in \{1, \ldots, m\}$. Let $n_i = \vert A_i \vert$ for $i\in \{1, \ldots, m\}$. Then, we have: all of the numbers $n_i$ (for ($i\in \{1, \ldots, m\}$) are even.
\end{corollary}

\begin{proof}
By Proposition \ref{prop:characterization2k}, the unique edge-disjoint leaf-to-leaf paths $P_1,\ldots P_k$ in $T$ and $P_1',\ldots, P_k'$ in $\widetilde{T}$, which exist by Lemma \ref{lem:2kuniquepaths}, share the same endpoints in $T$ and $\widetilde{T}$. Moreover, by the same argument as used in the proof of Lemma \ref{lem:2kuniquepaths}, $x$ and $y$ must be the two endpoints of one of the paths $P_i$ for some $i\in \{1,\ldots,k\}$ as there is no other way to turn both leaves of the cherry into endpoints  of $k$ edge-disjoint leaf-to-leaf paths in $T$, but they have to be such endpoints (as we have $k$ paths with two endpoints each, and as $n=2k$). Thus, by Proposition \ref{prop:characterization2k}, there must be a value $i\in \{1,\ldots,k\}$ with $P = P_i'$. 
Now, applying Corollary \ref{cor:even_subtrees} to $\widetilde{T}$ and $P\in \mathcal{P}:=\{P_1',\ldots, P_k'\}$ as the the unique set of $k$ edge-disjoint leaf-to-leaf paths in $\widetilde{T}$ given by Lemma \ref{lem:2kuniquepaths}, we get the desired result, which completes the proof.
\end{proof}

We require one more technical lemma concerning the case $n=2k$.

\begin{lemma}\label{lem:oneleafmoves} Let $k \in \mathbb{N}_{\geq 2}$. Let $X$ be a taxon set with $n=\lvert X \rvert =2k$ and $x,y \in X$, $x \neq y$. Let $S$ be a phylogenetic $X \setminus \{x,y\}$-tree. Let $T$ result from $S$ by choosing two edges $e,f \in E(S)$, $e \neq f$, and attaching leaf $x$ to $e$ by introducing a new vertex $v_x$ on $e$ as well as adding edge $\{v_x,x\}$ and attaching leaf $y$ to $f$ accordingly, introducing a new vertex $v_y$ on $f$. Let $P=x,v_x,\beta_1,\ldots,\beta_m,v_y,y$  denote the unique path from $x$ to $y$ in $T$. Let $T'$ result from $T$ by deleting $y$ and suppressing $v_y$ and re-attaching $y$ to any of the edges $\{x,v_x\}$, $\{v_x,\beta_1\}$, $\{\beta_1,\beta_2\}, \ldots, \{\beta_{m-1},\beta_m\}$. Moreover, let $S_{A_i}$ with taxon set $A_i$ and $\lvert A_i\rvert =n_i$ (with $i=0,\ldots,m+1$) be the set of trees resulting from the deletion of $P$ from $T$, cf. \ref{fig:lemoneleafmoves}. Then, if all $n_i$ are even, we have that $A_k(T)=A_k(T')$ and $T'$ can be obtained from $T$ by a series of $k$-problematic moves.
\end{lemma}

\begin{figure}
    \centering

\includegraphics{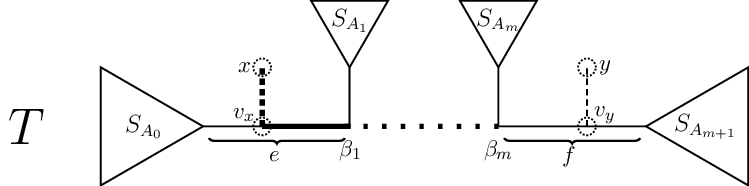}

\caption{Tree $T$ as described in Lemma \ref{lem:oneleafmoves}. When $x$ and $y$ get deleted, we obtain tree $S$. The bold path consists of the edges to which $y$ can be re-attached to form $T'$ (this includes the edge leading to $x$ (dashed), the edge from $v_x$ to $\beta_1$ and the entire path from $\beta_1$ to $\beta_m$ (dotted)).}
\label{fig:lemoneleafmoves}
\end{figure}

\begin{proof} We begin by noting that both $T$ and $T'$ can be obtained from $S$ by first attaching leaf $x$ to edge $e$ and subsequently attaching leaf $y$ (in case of $T'$, possibly on the edge leading to $x$). However, the attachment of $y$ to obtain  $T'$ is done such that $y$ is not added to any of the subtrees $S_{A_i}$ for $i=0,\ldots,m+1$, cf. Figure \ref{fig:lemoneleafmoves}. Thus, $T'$ can be reached from $T$ by a series of NNI moves by swapping the subtree consisting only of leaf $y$ with $S_{A_m},S_{A_{m-1}},\ldots$ until $T'$ is reached (note that if the last move swaps leaf $y$ with $S_{A_0}$, $T'$ contains the cherry $[x,y]$). We now argue that each of these NNI moves is $k$-problematic. 

Let $\mathcal{P}$ be the unique set of $k-1$ edge-disjoint leaf-to-leaf paths in $S$ (which must exist by Lemma \ref{lem:leafchoice} and which is unique by Lemma \ref{lem:2kuniquepaths} as $S$ has $2k-2=2(k-1)$ leaves). As all subtrees $S_{A_i}$ by assumption have an even number $n_i$ of leaves, none of the paths in $\mathcal{P}$ uses any of the edges that we used for the above mentioned NNI moves. This is due to the fact that if one leaf in, say, $S_{A_i}$ got connected to a leaf in $S_{A_j}$ for some $i\neq j$, then the remaining odd number of leaves in $S_{A_i}$ could not get connected in an edge-disjoint manner (either at least one leaf would remain unmatched or there would be two paths using the edge on which $S_{A_i}$ is pending). However, this shows that adding the unique path from $x$ to $y$ in $T$ and  $T'$ (which contains only edges that are not employed by any path in $\mathcal{P}$) to the set $\mathcal{P}$ of paths gives two sets $\mathcal{Q}$ and $\mathcal{Q}'$ of edge-disjoint leaf-to-leaf paths in $T$ and $T'$, respectively. However, as $\mathcal{Q}$ and $\mathcal{Q}'$ share the same set of endpoint pairs (namely $\{x,y\}$ as well as all endpoint pairs induced by $\mathcal{P}$), by Proposition \ref{prop:characterization2k}, we have $A_k(T)=A_k(T')$. By the same argument, we even have $A_k(T)=A_k(T'')$ for any tree $T''$ contained in the above described NNI-path from $T$ to $T'$. Thus, by Proposition \ref{prop:nni}, we know that $T'$ can be obtained from $T$ by a series of $k$-problematic moves. This completes the proof.
\end{proof}

Now that we have gained deeper understanding of the case $n=2k$, we can turn our attention again to the general case $n=2k+a$ with $a \in \{0,1,2\}$.

\subsubsection{Proof of Theorem \ref{thm:problematicmoves}}

Before we can proceed with the proof of Theorem \ref{thm:problematicmoves}, we first prove it for the special case in which two trees $T$ and $\widetilde{T}$ are almost identical except for one cherry being attached to different edges. This will be done by the following lemma, whose proof is quite technical. However, we will subsequently use this special case to prove the more general case.

\begin{figure}
    \centering
\includegraphics{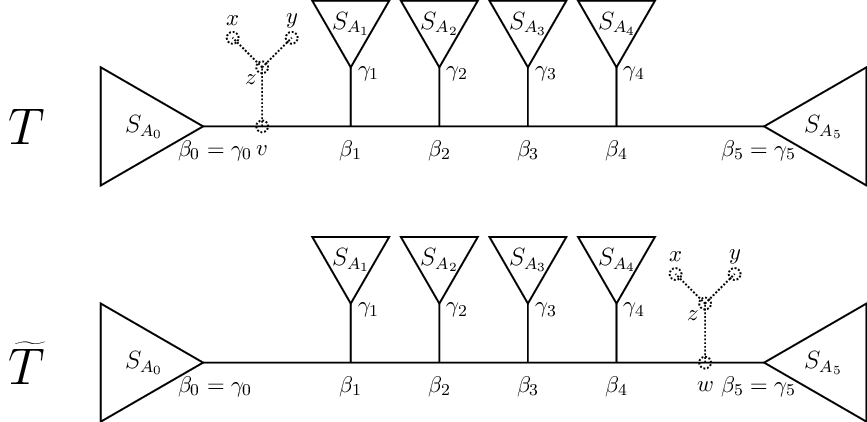}
\caption{$T$ and $\widetilde{T}$ as in Lemma \ref{lem:SPR_for_cherry} for $m=4$. Deleting the edges indicated by dotted lines will lead to the same tree $S$ in each case. }
\label{fig:ProofTheorem1}
\end{figure}

\begin{lemma}
    \label{lem:SPR_for_cherry}
    Let $k \in \mathbb{N}_{\geq 2}$. Let $n=2k+a$ with $a \in \{0,1,2\}$. Let $X$ be a set of taxa such that $\vert X \vert =n=2k+a$, and let $x,y \in X$. Moreover, let $S$ be a binary phylogenetic $X\setminus \{x, y\}$-tree. Now, let $T$ be the binary phylogenetic $X$-tree obtained from $S$ by attaching $[x,y]$ to an edge $e_1\in E(S)$, and let $\widetilde{T}$ be the binary phylogenetic $X$-tree obtained from $S$ by attaching $[x,y]$ to an edge $e_2\in E(S)$.
    Let $P = \beta_0, \beta_1, \ldots, \beta_m ,\beta_{m+1}$ be the unique path in $S$ with $e_1 = \{\beta_0, \beta_1\}$ and $e_2 = \{\beta_m, \beta_{m+1}\}$. Let $S_{A_0}, S_{A_1}, \ldots, S_{A_m}, S_{A_{m+1}}$ be the rooted binary phylogenetic trees resulting from deleting the internal vertices of $P$ from $S$ such that each $S_{A_i}$ has taxon set $A_i$ for $i\in \{0, \ldots, m+1\}$, cf. Figure \ref{fig:ProofTheorem1}. Let $n_i = \vert A_i \vert$ for $i\in \{0, \ldots, m+1\}$. Then, we have: 
    \begin{itemize}
        \item If $n=2k$, we have $A_k(T) = A_k(\widetilde{T})$.
        \item If $n=2k+1$, we have $A_k(T) = A_k(\widetilde{T})$ if and only if at most one of the numbers $n_i$ (with $i\in \{0, \ldots, m+1\}$) is odd. 

        \item If $n=2k+2$, we have $A_k(T) = A_k(\widetilde{T})$ if and only if either $m=0$ or none of the numbers $n_i$ (with $i\in \{0, \ldots, m+1\}$)  is odd.
    \end{itemize}
    Moreover, in each case we have $A_k(T) = A_k(\widetilde{T})$ if and only if $\widetilde{T}$ is obtained from $T$ by a series of problematic moves.
\end{lemma}

\begin{figure}
    \centering
\includegraphics{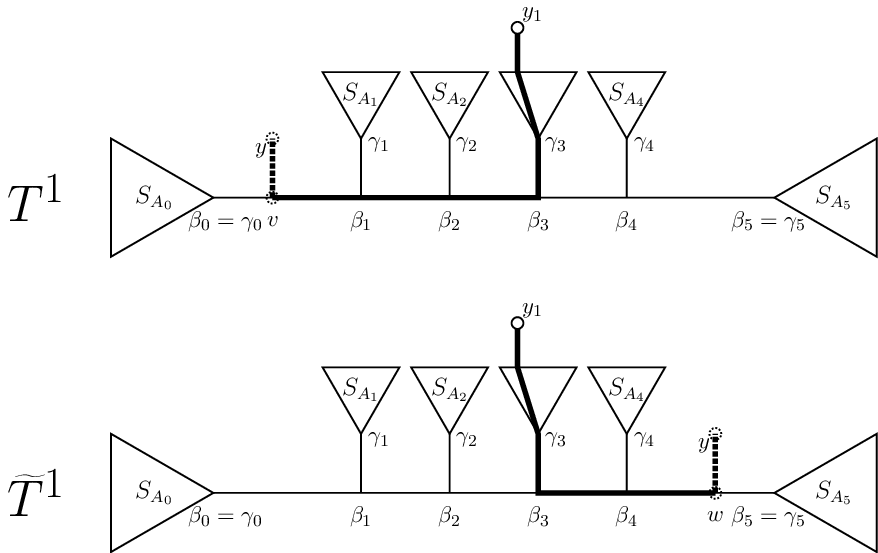}
\caption{This figure shows $T^1$ and $\widetilde{T}^1$ in the analysis of case $n=2k+1$ in the second part of the proof of Lemma \ref{lem:SPR_for_cherry} for $m=4$ and $i^{\ast} =3$. The thick lines in the upper image indicate path $Q$ and the thick lines in the lower image indicate path $Q'$. If we delete the internal vertices of $Q$ from $T^1$, we obtain a forest which contains, among others, $S_{A_0}, S_{A_1}$ and $S_{A_2}$ as components. Similarly, if we delete the internal vertices of $Q'$ from $\widetilde{T}^1$, then a forest is obtained which contains, among others, $S_{A_4}$ and $S_{A_5}$ as components.}
\label{fig:ProofTheorem2}
\end{figure}

\begin{proof}
    As a preparation for the proof we introduce some notation (cf. Figure \ref{fig:ProofTheorem1}).  Let $\gamma_i$ be the root of $S_{A_i}$ for $i\in \{0, \ldots, m+1\}$. This immediately leads to $\beta_0=\gamma_0$ and $\beta_{m+1}=\gamma_{m+1}$. Moreover, without loss of generality we can assume that $S_{A_i}$ is adjacent to $\beta_{i}$ in $S$ for $i=1,\ldots,m$, which implies that $\{\gamma_i, \beta_i\} \in E(S)$ for $i\in \{1, \ldots, m\}$. Next, let $v\in V(T)$ be the vertex resulting from subdividing $e_1$ such that $\{v,z\}\in E(T)$, where $z$ is the common neighbor of $x$ and $y$. Similarly, let $w\in V(\widetilde{T})$ be the vertex resulting from subdividing $e_2$ such that $\{w,z\}\in E(\widetilde{T})$.

We start by observing that $T$ can clearly be transformed into $\widetilde{T}$ by NNI moves as follows (cf. Figure \ref{fig:ProofTheorem1}): We first perform a move around the edge $\{v,\beta_1\}$; in particular, we swap the cherry $[x,y]$ with subtree $S_{A_1}$. We call the resulting tree $T^{NNI}_{A_1}$. Subsequently, we swap the cherry $[x,y]$ with subtree $S_{A_2}$ to derive tree $T^{NNI}_{A_2}$. We repeat this procedure until we reach $T^{NNI}_{A_m}=\widetilde{T}$. We consider the NNI path $T,T^{NNI}_{A_1},\ldots,T^{NNI}_{A_m}=\widetilde{T}$. The $i^{th}$ NNI move in this path can be described by defining trees $T_1^i$ with leaf set $\{x,y\}$, $T_2^i$ with leaf set $\bigcup_{j=0}^{i-1}A_{j}$, $T_3^i$ with leaf set $A_i$ and $T_4^i$ with leaf set $\bigcup_{j=i+1}^{m+1}A_{j}$, where the NNI move then swaps $T_1^i$ and $T_3^i$. Consequently, the number of leaves in these four subtrees affected by the $i^{th}$ NNI move are $n_1^i=2$ for $T_1^i$, $n_2^i=\sum_{j=0}^{i-1}n_i$ for $T_2^i$, $n_3^i=n_i$ for $T_3^i$, and $n_4^i=\sum_{j=i+1}^{m+1}n_i$ for $T_4^i$.

Note that for all $i=1,\ldots, m$,  $T^{NNI}_{A_i}$ can alternatively be derived from $S$ by attaching cherry $[x,y]$ at an appropriate position.
    
    \begin{enumerate}
    \item We first  consider the case $n=2k$. 
     As we have $\vert X \setminus \{x,y\} \vert = 2k-2=2(k-1)$, we know by Lemma \ref{lem:2kuniquepaths} that there is a unique set $\mathcal{P}=\{P_1,\ldots,P_{k-1}\}$ of edge-disjoint leaf-to-leaf paths contained in $S$. Moreover, adding the path $P=x,z,y$, where $z$ is the common neighbor of $x$ and $y$, to $\mathcal{P}$ clearly gives a set of edge-disjoint leaf-to-leaf paths for $T$ and $T^{NNI}_{A_i}$ (with $i=1,\ldots, m$), including $\widetilde{T}=T^{NNI}_{A_m}$. By Lemma \ref{lem:2kuniquepaths}, this set of paths $\mathcal{P} \cup \{P\}$ is a unique set of edge disjoint leaf-to-leaf paths for each tree under consideration. Moreover, for each such tree, these paths have the same sets of endpoint pairs. By Proposition \ref{prop:characterization2k}, this implies $A_k(T)=A_k\left(T^{NNI}_{A_i} \right)=A_k(\widetilde{T})$ for all $i=1,\ldots,m$. Additionally, using Proposition \ref{prop:nni}, this shows (as $T,T^{NNI}_{A_1},\ldots, T^{NNI}_{A_{m-1}},\widetilde{T}$ forms an NNI path from $T$ to $\widetilde{T}$) that $\widetilde{T}$ can be derived from $T$ by a series of $k$-problematic moves. This completes the proof for the case $n=2k$.
    
    \item So from now on, we may assume $n=2k+a$ with $a\in \{1,2\}$. We start by assuming that one of the following statements is true: 
    \begin{itemize}
        \item $n=2k+1$ and at most one of the numbers $n_i$ (for $i\in \{0, \ldots, m+1\}$) is odd.
        \item $n=2k+2$ and ($m=0$ or none of the numbers $n_i$ (for $i\in \{0, \ldots, m+1\}$) is odd).
    \end{itemize}
    In each of these cases, we need to show $A_k(T) = A_k(\widetilde{T})$ and that $\widetilde{T}$ is obtained from $T$ by a series of $k$-problematic moves. We will do so by showing that the above mentioned NNI path $T, T^{NNI}_{A_1},\ldots,T^{NNI}_{A_m}=\widetilde{T}$ is $k$-problematic.
  
As explained above, the $i^{th}$ NNI move in this path from $T$ to $\widetilde{T}$ can be described by considering trees $T_1^i$, $T_2^i$, $T_3^i$ and $T_4^i$ with $n_1^i=2$, $n_2^i=\sum_{j=0}^{i-1}n_i$, $n_3^i=n_i$, and $n_4^i=\sum_{j=i+1}^{m+1}n_i$ leaves, respectively, and then swapping $T_1^i$ and $T_3^i$. Using a simple parity argument, it can be easily seen that if at most $2-a$ of the values $n_0, n_1,\ldots, n_{m+1}$ are odd (for $a\in \{1,2\}$), then also at most $2-a$ of the values $n_1^i,n_2^i, n_3^i, n_4^i$ are odd. Using Corollary \ref{cor:cases}, we conclude that $A_k(T)=A_k\left(T^{NNI}_{A_i}\right)=A_k(\widetilde{T})$ for all $i=1,\ldots,m+1$. Thus, by Proposition \ref{prop:nni}, $\widetilde{T}$ can indeed be derived from $T$ by a series of $k$-problematic moves, and we have $A_k(T) = A_k(\widetilde{T})$ as desired. This completes the second part of the proof. 
    
    \item We continue with the third part of the proof. In this part, we want to show that $n=2k+a$ with $a\in \{1,2\}$ and $ m>0$ together with $A_k(T) = A_k(\widetilde{T})$ imply that at most $2-a$ of the numbers $n_i$ (with $i\in \{0, \ldots, m+1\}$) are  odd. We will first prove the statement for $a=1$ and subsequently reduce the case $a=2$ to this case.

\begin{enumerate}
\item We consider the case $n=2k+1$. Now, as $A_k(T) = A_k(\widetilde{T})$, we know from Corollary \ref{cor:cherryred} that $A_k(T^1) = A_k(\widetilde{T}^1)$, where $T^1$ and $\widetilde{T}^1$ are obtained by cherry reductions of type 1 using the cherry $[x,y]$. Note that $T^1$ and $\widetilde{T}^1$ are binary phylogenetic $X \setminus \{x\}$-trees with $\lvert X\setminus \{x\} \rvert = 2k$. 
    Therefore, by Lemma \ref{lem:leafchoice} we know that $T^1$ contains a set of $k$ edge-disjoint leaf-to-leaf paths $\mathcal{P}=\{P_1, \ldots, P_k\}$ and that $\widetilde{T}^1$ contains a set of $k$ edge-disjoint leaf-to-leaf paths $\mathcal{P}'=\{P'_1, \ldots, P'_k\}$. Each $z\in X\setminus \{x\}$ is an endpoint of exactly one element of $\mathcal{P}$ and of exactly one element of $\mathcal{P}'$. In particular, $y$ is an endpoint of some $Q \in \mathcal{P}$ and of some $Q'\in \mathcal{P}'$. Let $y_1$ be the other endpoint of $Q$ and let $y_2$ be the other endpoint of $Q'$, as depicted by Figure \ref{fig:ProofTheorem2}, respectively. Now, by Proposition \ref{prop:characterization2k}, we know from $A_k(T^1)=A_k(\widetilde{T}^1)$ that the pairs of endpoints of $\mathcal{P}$ coincide with the ones $\mathcal{P}'$. In particular, this shows that $y_1=y_2$.
    
    Next, we consider the rooted subtrees $T_1, \ldots, T_g$ obtained by deleting the vertices of $Q$ from $T^1$ such that each $T_i$ has taxon set $X_i \subset X\setminus \{x\}$ for $i\in \{1, \ldots, g\}$. Similarly, we consider the rooted subtrees $T'_1, \ldots, T'_h$ obtained by deleting the vertices of $Q'$ from $\widetilde{T}^1$ such that each $T'_i$ has taxon set $X'_i \subset X\setminus \{x\}$ for $i\in \{1, \ldots, h\}$. Applying Corollary \ref{cor:even_subtrees} to $T^1$ and $Q\in \mathcal{P}$, we can conclude that $\lvert X_i \rvert$ is even for all $i\in \{1, \ldots, g\}$. Similarly, applying  Corollary \ref{cor:even_subtrees} to $\widetilde{T}^1$ and $Q'\in \mathcal{P}'$,  we conclude that $\lvert X'_j\lvert$ is even for $j\in \{1, \ldots, h\}$. 
    
    Note that there is some $i^{\ast}\in \{0, 1, \ldots, m+1\}$ with $y_1\in A_{i^{\ast}}$. Now, if $i^{\ast} > 0$, then $Q$ contains the subpath $y, v, \beta_1, \ldots, \beta_{i^{\ast}}$. On the other hand, if $i^{\ast} < m+1$, then $Q'$ contains the subpath $y, w, \beta_m, \ldots, \beta_{i^{\ast}}$. Now, in the first case, $v, \beta_1, \ldots, \beta_{i^{\ast}-1}$ are internal vertices of $T$ adjacent to $\gamma_0, \ldots, \gamma_{i^{\ast}_1}$, which are not contained in $Q$. Then, deleting the vertices $y, v, \beta_1, \ldots, \beta_{i^{\ast}}$ leads to a forest containing, amongst others, the rooted subtrees $S_{A_0}, \ldots, S_{A_{i^{\ast}-1}}$ as components. However, this implies that the sets $A_i$ (for all $i: 0\leq i \leq i^{\ast}-1$) are contained in the collection of the sets $X_i$ (with $i: 1\leq i\leq g$). Similarly, in the second case we see that deleting vertices $y, w, \beta_m, \ldots, \beta_{i^{\ast}}$ from $\widetilde{T}$ leads to a forest containing, amongst others, the rooted subtrees $S_{A_{i^{\ast}+1}}, \ldots, S_{A_{m+1}}$. Thus, the sets $A_i$ (for all $i: i^{\ast} + 1\leq i \leq  m+1$) are contained in the collection of the sets $X'_i$ (with $i: 1\leq i\leq h$). 
    We introduce the notation $\mathcal{X}:=\{X_i: 1\leq i \leq g\} \cup \{X_i': 1 \leq i \leq h\}$ and conclude from our observations that for all $i\in \{0,\ldots, m+1\}$ with $i \neq i^\ast$ we have $A_i \in \mathcal{X}$.
   
As we have already shown above that all elements of $\mathcal{X}$ are of even cardinality, this completes the proof that at most one of the numbers $n_i$ for ($i\in \{0, \ldots, m+1\}$) is odd if $n=2k+1$.

   \item Now, we consider the only remaining case, namely $n=2k+2$. 
    Seeking a contradiction, we assume that $m>0$, that there is an odd number $n_{i^\ast}$ with $i^\ast \in \{0, \ldots, m+1\}$, and that $A_k(T) = A_k(\widetilde{T})$.
We consider two cases. 
    \begin{itemize}
        \item In case we have $n_i = 1$ for all $\{0, \ldots, m+1\}$, we know in particular that $n_0 = 1$ and $A_0 = \{\beta_0\}$. Let $T^1$ and $\widetilde{T}^1$ be the trees resulting from $T$ and $\widetilde{T}$ by a cherry reduction of type 1 using $[x,y]$. Then, $[y,\beta_0]$ is a cherry in $T^1$ but not in $\widetilde{T}^1$. As $T^1$ and $\widetilde{T}^1$ are binary phylogenetic $X\setminus \{x\}$-trees and $\vert X\setminus \{x\}\vert = n-1 = 2k+1$, we can apply Proposition \ref{prop:cherry} and conclude that $A_k(T^1) \neq A_k(\widetilde{T}^1)$. By Corollary \ref{cor:cherryred} it follows that $A_k(T) \neq A_k(\widetilde{T})$, a contradiction to our assumption. This shows that the case that $n_i=1$ for all $i=0,\ldots,m+1$ cannot hold.

        \item Now we consider the case in which there exists a value of $j\in \{0, \ldots, m+1\}$ such that $ n_j > 1$. In the following, we will fix a specific such value of $j$ (with certain properties which we will determine subsequently). But regardless of the value of $j$ with $n_j>1$, we can do the following construction  (whose outcome will depend on the specific choice of $j$ and $n_j$): As $n_j>1$, $S_{A_j}$ contains a cherry 
        $[a,b]$.  Note that $[a,b]$ must also be a cherry in $T$, $ \widetilde{T}$ and $S$. Let $T^1$, $\widetilde{T}^1$ and  $S^1$, respectively, denote the trees we obtain from a cherry reduction of type 1 using $[a,b]$. Then, $T^1$ and $ \widetilde{T}^1$ are binary phylogenetic $X\setminus \{a\}$-trees.  Note that $P = \beta_0, \beta_1, \ldots, \beta_m ,\beta_{m+1}$ is a path in $S^1$ which contains $e_1$ and $e_2$, and that $T^1$ is obtained from $S^1$ by attaching cherry $[x,y]$ to $e_1$ and that $\widetilde{T}^1$ is obtained from $S^1$ by attaching cherry $[x,y]$ to $e_2$. 
        
        Next, recall that Lemma \ref{lem:SPR_for_cherry} has already been proven for the case $n=2k+1$. Therefore, we can apply the lemma to $T^1$, $\widetilde{T}^1$ and  $S^1$. Then, let $S'_{A'_i}$ be the rooted subtrees of $S^1$ obtained from deleting the vertices of $P$ from $S^1$ such that each $S'_{A'_i}$ has taxon set $A'_i$ for all $i$ with $0\leq i\leq m'+1$. Let $n'_i = \lvert A'_i\rvert$ (for $0\leq i\leq m'+1$). It can easily be seen that we have $m' = m$, $n'_i = n_i$ (for all $i\neq j$ with $ 0\leq i\leq m+1$), and that $n'_j = n_j -1$ (where the latter is true as taxon $a$ has been deleted from $S$ to derive $S^1$). 
        
    We now consider two subcases. 
        \begin{itemize}\item Suppose $j$ can be chosen such that $n_j>1$ is even. Then we fix such a value of $n_j$.  
    By Corollary \ref{cor:cherryred} we know that $A_k(T) = A_k(\widetilde{T})$ implies $A_k(T^1) = A_k(\widetilde{T}^1)$. By Lemma \ref{lem:SPR_for_cherry} applied to the case $n=2k+1$, which we have already proven, we know that at most one of the numbers $n'_i$ (with $i: 0\leq i\leq m+1$) is odd. But we have assumed that there exists some $i^{\ast}\in \{0, \ldots, m+1\}$ such that $n_{i^{\ast}}$ is odd. As  $n_j$ was chosen to be even, this implies  $j\neq i^{\ast}$, and we conclude that $n'_{i^{\ast}}$ and $n'_j=n_j-1$ are both odd. However, this contradicts the case $n=2k+1$ of Lemma \ref{lem:SPR_for_cherry}, by which at most one of these values can be odd.
        \item Suppose there is no $j$ such that $n_j>1$ is even. In this case, we fix any of the values $n_j>1$. 
        Then, clearly all values of $n_i$ with $i\in \{0, \ldots, m+1\}$ are odd (as they are either 1 or, in case they are larger than 1, still cannot be chosen to be even). But as we are in the case where $m>0$, this implies that we have at least three $n_i$'s (with $i: 0\leq i\leq m+1$) which are odd, namely $n_0$, $n_m$ and $n_{m+1}$. However, the numbers $n'_i$ (with $i: 0\leq i\leq m+1$) are identical with the numbers $n_i$ with precisely one exception. But then, clearly at least two of the numbers $n'_i$ (with $i: 0\leq i\leq m+1$) are odd.  Again this is a contradiction to the case $n=2k+1$ of Lemma \ref{lem:SPR_for_cherry}, by which at most one of these values can be odd.
        \end{itemize}

So both subcases lead to a contradiction, which shows that the assumption that there exists a value of $j$ with $n_j>1$ cannot hold. 
\end{itemize}

   As all cases lead to contradictions, it is clear that our assumption was wrong, i.e., we cannot have $m>0$ and  $A_k(T)=A_k(\widetilde{T})$ if there is an odd number $n_{i^\ast}$ with $i^\ast \in \{0,\ldots,m+1\}$. This proves the case $n=2k+2$. As we have already proven the case $n=2k+1$, this completes the proof of the third part.
    \end{enumerate}
    \end{enumerate}

    Last but not least, we note that in all cases, we have $A_k(T)=A_k(\widetilde{T})$ if and only if $T$ and $\widetilde{T}$ are connected via a series of $k$-problematic moves. For $n=2k$, we have shown this equivalence explicitly in the first part of the proof. For $n=2k+a$ with $a \in \{1,2\}$, we have shown in the third part of the proof that if $A_k(T)=A_k(\widetilde{T})$, we have that at most $2-a$ of the $n_i$-values are odd, which by the second part of the proof implies that $T$ and $\widetilde{T}$ are connected by a series of $k$-problematic moves. On the other hand, if $\widetilde{T}$ can be reached from $T$ by a series of such moves, we know by Proposition \ref{prop:nni} that $A_k(T)=A_k(\widetilde{T})$. 
    This completes the proof.
\end{proof}

Now we are finally in a position  to prove Theorem \ref{thm:problematicmoves}.

\begin{proof}
    If $\widetilde{T}$ is obtained from $T$ by a series of problematic moves, then by Proposition \ref{prop:nni} we know that $A_k(T) = A_k(\widetilde{T})$. So one direction of the theorem is clear. For the other direction we have to show that $A_k(T) = A_k(\widetilde{T})$ implies that $\widetilde{T}$ is obtained from $T$ by a series of problematic moves. We have already mentioned that we know the theorem to be true for $n<2k$ and $n>2k+2$. So it suffices to consider the cases $n=2k+a$ for $a\in \{0,1,2\}$.

    For our proof we use induction over $k$. By Proposition \ref{prop:A1A2} we know that the $A_2(T) = A_2(\widetilde{T})$ implies $T\cong \widetilde{T}$. By definition of a series of $k$-problematic moves, this shows that the theorem holds for $k=2$ (as the series in this case is empty). Hence, we may subsequently assume that we have $k>2$ and that the theorem holds for $k-1$.

We distinguish two cases.

\begin{enumerate}
\item At first we consider the case that $T$ and $\widetilde{T}$ have a cherry $[x,y]$ in common. Let $T^2$ and $\widetilde{T}^2$ be the binary phylogenetic trees obtained by a cherry reduction of type 2 using $[x,y]$ from $T$ and $\widetilde{T}$, respectively. From $A_k(T) = A_k(\widetilde{T})$ we conclude by Corollary \ref{cor:cherryred} that $A_{k-1}(T^2) = A_{k-1}(\widetilde{T}^2)$. By induction, it follows that $\widetilde{T}^2$ is obtained from $T^2$ by a series of $k-1$-problematic moves.
    
    Now, remember that the inverse operation to a cherry reduction of type 2 consists of attaching a cherry $[x,y]$ to an edge. So there exists some $e_1\in E(T^2)$ such that we can derive $T$ from $T^2$ by attaching $[x,y]$ to $e_1$. Hence, we can apply Lemma \ref{lem:series_problematic_moves} and conclude that there exists some $e_2\in E(\widetilde{T}^2)$ and some binary phylogenetic $X$-tree $T^{\ast}$ obtained by attaching $[x,y]$ to $e_2$ such that $T^{\ast}$ is obtained from $T$ by a series of $k$-problematic moves. Moreover, by Proposition \ref{prop:nni}, this implies $A_k(T)=A_k\left(T^\ast\right)$. In particular, for all three trees $T$, $\widetilde{T}$ and $T^\ast$, we now know that their $A_k$-alignments are equal, i.e., we have $A_k(\widetilde{T})=A_k(T)=A_k\left(T^\ast\right)$.
    
    It remains to show that $\widetilde{T}$ is obtained from $T^{\ast}$ by a series of $k$-problematic moves (as we can then form a series of $k$-problematic moves from $T$ to $\widetilde{T}$ via $T^\ast$). Again as attaching a cherry is the inverse operation to a cherry reduction, we know that there exists some $e_3\in E(\widetilde{T}^2)$ such that $\widetilde{T}$ is obtained from $\widetilde{T}^2$ by attaching $[x,y]$ to $e_3$. So we know that  both $T^{\ast}$ and $\widetilde{T}$ can be constructed by attaching $[x,y]$ to some edge of $\widetilde{T}^2$, and, as explained above, we have $A_k(\widetilde{T})=A_k\left(T^\ast\right)$. That means that we can apply Lemma \ref{lem:SPR_for_cherry} and see that $A_k(T^{\ast}) = A_k(\widetilde{T})$ implies that $\widetilde{T}$ is obtained from $T^{\ast}$ by a series of $k$-problematic moves. As we have already proven that $T^{\ast}$ can be obtained from $T$ by a series of problematic moves, in total we get the desired result that $\widetilde{T}$ is obtained from $T$ by a series of problematic moves. This completes the first part of the proof.

    \item Now we assume that there is no cherry contained in both $T$ and $\widetilde{T}$. By Proposition \ref{prop:cherry} this is not compatible with our assumption $A_k(T) = A_k(\widetilde{T})$ except if $n=2k$. So it suffices to consider the case $n=2k$. Now $T$ contains at least one cherry $[x,y]$ as $k\geq 2$ and thus $n\geq 4$. Let $S$ be the tree resulting from $\widetilde{T}$ by deleting $x$ and $y$ and suppressing the resulting degree-2 vertices. By Corollary \ref{lem:cherry_n=2k}, we know from $A_k(T)=A_k(\widetilde{T})$ and the fact that $[x,y]$ is a cherry in $T$, that the subtrees $T_{A_i}$ we derive from $\widetilde{T}$ when removing the unique $x$-$y$-path from this tree all have an even number $n_i$ of leaves.
    
    Moreover, as $[x,y]$ is not a cherry in $\widetilde{T}$, there are $e_1, e_2\in E(S)$ with $e_1\neq e_2$ such that $\widetilde{T}$ is obtained from $S$ by attaching $x$ to $e_1$ and attaching $y$ to $e_2$. Let $T^{\ast}$ be the tree obtained from $S$ by attaching $[x,y]$ to $e_1$. Then, as all the $n_i$ are even as explained above, we can use Lemma \ref{lem:oneleafmoves} and conclude that  $T^\ast$ can be obtained from $\widetilde{T}$ by a series of $k$-problematic moves, and $A_k(T^{\ast})=A_k(\widetilde{T})$.
Using the assumption that $A_k(\widetilde{T}) = A_k(T)$,  we conclude that $A_k(T^{\ast}) = A_k(T)$. 

However, note that $T$ and $T^{\ast}$ have cherry $[x,y]$ in common, which means we can apply the first part of the proof to these two trees. We conclude that $T^{\ast}$ can be obtained from $T$ by a series of problematic moves. Together with the fact that $\widetilde{T}$ is obtained from $T^{\ast}$ by a series of problematic moves as explained, we conclude that $\widetilde{T}$ is obtained from $T$ by a series of problematic moves (namely via $T^{\ast}$), which completes the proof. 
    \end{enumerate}

\end{proof}

\section{Discussion and outlook}

In the present manuscript we have completely characterized all pairs of binary phylogenetic $X$-trees with identical $A_k$-alignments. For $|X|=n\geq 2k+3$, we have shown that two trees have identical $A_k$-alignments if and only if they are isomorphic. In case that $n<2k$, it is easy to see that $A_k(T) = \emptyset$ for all binary phylogenetic $X$-trees $T$. The three cases of $n\in \{2k,2k+1,2k+2\}$ are most involved as it was already shown in a previous manuscript that there exist pairs of non-isomorphic binary phylogenetic $X$-tree with non-identical $A_k$-alignments alongside pairs with identical $A_k$-alignments \cite[Proposition 4]{WildeFischer2023}.

The main contribution of the present manuscript is the introduction of a class of tree operations guaranteeing that for a given binary phylogenetic $X$-tree $T$, all binary phylogenetic $X$-trees $\widetilde{T}$ with $A_k(\widetilde{T}) = A_k(T)$ can be obtained by a composition of these operations, cf. Theorem \ref{thm:problematicmoves}. As this theorem contains a necessary as well as sufficient condition for the identity of $A_k$-alignments, it provides a complete characterization of all pairs of binary phylogenetic $X$-trees with identical $A_k$-alignments.

As the gap in the literature concerning cases in which the $A_k$-alignment uniquely characterizes a tree is now closed, future research can focus on aspects of when these unique trees can also be recovered by maximum parsimony when used as a tree reconstruction criterion. In fact, it is so far only known for $k\leq 2$ that such a reconstruction is possible \cite{Fischer2019} as well as for the restricted case of considering only trees in the so-called NNI neighborhood of the original tree \cite{fischer2024}. Now that it is clear when the necessary condition of the uniqueness of the $A_k$-alignment holds, the tree reconstruction question can finally be tackled. 

Another area for future research might be the generalization of our results for binary trees to non-binary ones. Some progress in this regard has very recently been made in \cite{Qian2024}.

\section*{Conflict of interest}
The authors state that there is no conflict of interest.

\section*{Data availability statement}
No data was used or generated within the scope of this study. 

\bibliographystyle{plainnat}
\bibliography{references}   

\end{document}